\documentclass[a4paper,onecolumn,11pt,accepted=2023-04-17]{quantumarticle}
\pdfoutput=1
\usepackage[utf8]{inputenc}
\usepackage[english]{babel}
\usepackage[T1]{fontenc}
\usepackage{hyperref}
\usepackage[numbers,sort&compress]{natbib}
\usepackage{qcircuit}
\usepackage{graphics}
\usepackage{amsmath,amssymb,amsthm,mathrsfs,amsfonts,dsfont}
\usepackage{subfigure, epsfig}
\usepackage{braket}
\usepackage{bm}
\usepackage{enumerate}
\usepackage{algorithm}
\usepackage{physics}
\usepackage{pgfplots}
\pgfplotsset{width=7cm,compat=1.9}
\usepgfplotslibrary{external}
\usepackage{orcidlink}

\usepackage{algpseudocode}
\usepackage{color}

\graphicspath{{./figure/}}

\newtheorem{theorem}{Theorem}

\newtheorem{lemma}{Lemma}
\newtheorem{corollary}{Corollary}
\newtheorem{prop}{Proposition}
\newtheorem{definition}{Definition}

\newcommand{\mc}{\mathcal}
\newcommand{\mb}{\mathbf}
\newcommand{\mbb}{\mathbb}
\newcommand{\id}{\mathbb{I}}
\newcommand{\sw}{\mathbb{S}}

\newcommand{\sym}{\mathrm{sym}}

\newcommand{\comments}[1]{}
\usepackage{tikz}
\usepackage{lipsum}

\begin{document}
\title{Performance analysis of multi-shot shadow estimation}

\author{You Zhou~\orcidlink{0000-0003-0886-077X}}
\email{you\_zhou@fudan.edu.cn}
\affiliation{Key Laboratory for Information Science of Electromagnetic Waves (Ministry of Education), Fudan University, Shanghai 200433, China}

\author{Qing Liu~\orcidlink{0000-0002-1576-1975}}
\email{liuqingppk@gmail.com}
\affiliation{Key Laboratory for Information Science of Electromagnetic Waves (Ministry of Education), Fudan University, Shanghai 200433, China}
\maketitle

\begin{abstract}
 Shadow estimation is an efficient method for predicting many observables of a quantum state with a statistical guarantee. In the multi-shot scenario, one performs projective measurement on the sequentially prepared state for $K$ times under the same unitary evolution, and repeats this procedure for $M$ rounds of random sampled unitary, which results in $MK$ times measurements in total. Here we analyze the performance of shadow estimation in this multi-shot scenario, which is characterized by the variance of estimating the expectation value of some observable $O$. We find that in addition to the shadow-norm $\|O \|_{\mathrm{shadow}}$ introduced in [Huang et.al.~Nat.~Phys.~2020\cite{huang2020predicting}], the variance is also related to another norm, and we denote it as the cross-shadow-norm $\|O \|_{\mathrm{Xshadow}}$. For both random Pauli and Clifford measurements, we analyze and show the upper bounds of $\|O \|_{\mathrm{Xshadow}}$. In particular, we figure out the exact variance formula of Pauli observables for random Pauli measurements. Our work gives theoretical guidance for the application of multi-shot shadow estimation. 
\end{abstract}

\section{Introduction}
Learning the properties of quantum systems is of fundamental and practical interest, which can uncover quantum physics and enable quantum technologies. Traditional quantum state tomography \cite{haah2017sample,flammia2012quantum,kliesch2021theory} is not efficient with the increasing of qubit number in various quantum simulating and computing platforms \cite{altman2021quantum,alexeev2021quantum}. Recently, randomized measurements \cite{elben2023randomized}, especially the shadow estimation method \cite{aaronson2019shadow,huang2020predicting} are proposed for learning the quantum systems efficiently with a statistical guarantee. 

In shadow estimation, a quantum state $\rho$ is prepared sequentially, evolved by a randomly sampled unitary $\rho\rightarrow U\rho U^{\dag}$, and finally measured in the computational basis. With the measurement results and the information of $U$, one can construct the shadow snapshot $\hat{\rho}$, which is an unbiased estimator of $\rho$. Using a few independent snapshots, many properties described by the observable $O$, for instance, the local Pauli operators and fidelities to some entangled states \cite{huang2020predicting}, can be predicted efficiently. Shadow estimation has been applied to many quantum information tasks, from quantum  correlation detection  \cite{elben2020mixedstate,rath2021Fisher,liu2022detecting}, quantum chaos diagnosis \cite{garcia2021quantum,mcginley2022quantifying}, to quantum error mitigation \cite{seif2023shadow,Hu2022Logical}, quantum machine learning \cite{huang2022quantum,huang2022provably}, and near-term quantum algorithms \cite{Stefan2022BP}.
There are also a few works aiming to enhance the performance of shadow estimation from various perspectives, for instance, derandomization \cite{huang2021efficient} and locally biased shadow for measuring Pauli observables \cite{hadfield2022measurements,wu2023overlapped}, hybrid shadow for polynomial functions \cite{zhou2022hybrid}, single-setting shadow \cite{Roman2022Single}, generalized-measurement shadow \cite{Chau2022Generalized}, shadow estimation with a shallow circuit \cite{hu2023classical,bertoni2022shallow,arienzo2022closed} or restricted unitary evolution \cite{Hu2022Hamiltonian,tran2023measuring,mcginley2022shadow,van2022hardware}, and practical issues considering measurement noises \cite{Chen2021Robust,koh2022classical}.

In the current framework, shadow estimation in general needs to sample say $M$ times random unitary to generate $M$ shadow snapshots. However, the execution of different unitary is still resource-consuming in real experiments, which generally means changing control pulses or even physical settings. One direct idea is to add more measurement shots under the same unitary evolution \cite{elben2023randomized,seif2023shadow}, say $K$-shot per unitary, to compensate for the realization of random unitaries. This strategy, denoted as multi-shot shadow estimation, however, still lacks a rigorous performance analysis compared to the original one. 

\begin{figure}[h]
\centering
\includegraphics[width=0.6\textwidth]{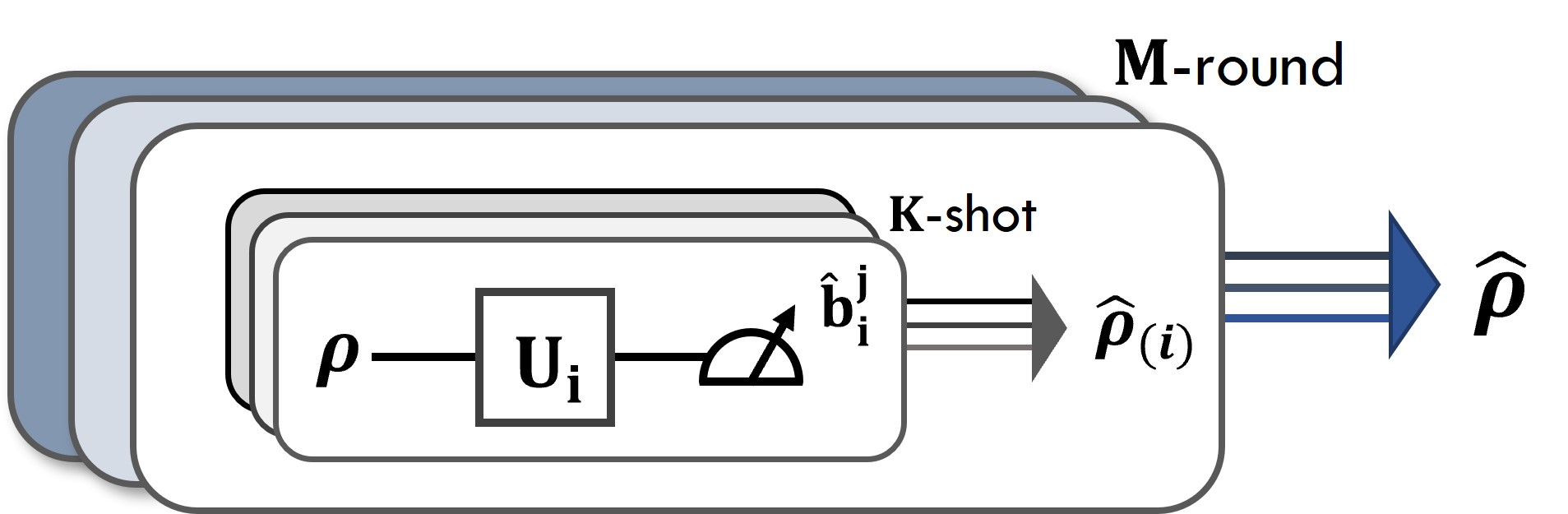}
\caption{Multi-shot shadow estimation. }
\label{fig:multi-shot}
\end{figure}

In this work, we fill this gap by giving the statistical guarantee for multi-shot shadow estimation. We estimate the variance of measuring some observable $O$ and relate it to the introduced cross-shadow-norm $\|O \|_{\mathrm{Xshadow}}$, which supplements the original shadow-norm $\|O \|_{\mathrm{shadow}}$ \cite{huang2020predicting}. The final variance is practically bounded by a interpolation of $\|O \|_{\mathrm{Xshadow}}$ and $\|O \|_{\mathrm{shadow}}$ controlled by the shot-number $K$. For both Pauli measurements and Clifford measurements, where the random unitary is sampled from single-qubit and global Clifford circuits respectively, we analyze $\|O \|_{\mathrm{Xshadow}}$ and the variance. In the Pauli measurement, we find that multi-shot shadow estimation shows an advantage when estimating Pauli observable, conditioning on the expectation value of that observable. However, in the Clifford measurement, we find no clear benefit of increasing the shot-number $K$ for a given number of random evolution $M$, especially for estimating the fidelity for many-qubit states. Our work could advance further applications of multi-shot shadow estimation.

\section{Multi-shot shadow estimation}

In this section, we introduce the background of shadow estimation and its multi-shot extension. Here, we focus on the $n$-qubit quantum system, that is, the state $\rho\in \mc{H}_D$ with $D=2^n$, and denote the computational basis of $\mc{H}_D$ as $\{\ket{\mb{b}}\}=\{\ket{b_1b_2\cdots b_n}\}$ with $b_i=0/1$. 
In shadow estimation \cite{huang2020predicting}, one prepares an unknown quantum state $\rho$ in the experiment sequentially for $M$ rounds. In the $i$-th round, one evolves the quantum system with a random unitary $U$ sampled from some ensemble $\mc{E}$ to get $U\rho U^{\dag}$, and measures it in the computational basis to get the result $\ket{\mb{b}^{(i)}}$.  We call $M$ the setting-number hereafter, since the sampled unitary $U$ in each round determines the measurement setting. The shadow snapshot can be constructed as follows.
\begin{equation}\label{eq:shadowPost}
    \begin{aligned}
       \hat{\rho}_{(i)}:= \mc{M}^{-1}\left(U^{\dag}\ket{\hat{\mb{b}}^{(i)}}\bra{\hat{\mb{b}}^{(i)}}U\right),
    \end{aligned}
\end{equation}
with $\hat{\mb{b}}^{(i)}$ a random variable. $\hat{\rho}_{(i)}$ is an unbiased estimator of $\rho$, such that $\mathbb{E}_{\{U,\mb{b}^{(i)}\}}\left(\hat{\rho}_{(i)}\right)=\rho$.  And the inverse (classical) post-processing $\mc{M}^{-1}$ is determined by the chosen random unitary ensemble \cite{huang2020predicting,hu2023classical,ohliger2013efficient,Hu2022Tensor}. For the $n$-qubit random Clifford circuit ensemble $\mc{E}_{\mathrm{Cl}}$ and the tensor-product of random single-qubit Clifford gate ensemble $\mc{E}_{\mathrm{Pauli}}$, one has $\mc{M}_C^{-1}=\mc{M}_n^{-1}$ and $\mc{M}_P^{-1}=\otimes_{i=1}^n\mc{M}_1^{-1}$ respectively, with $\mc{M}_n^{-1}(A)=(2^n+1)A-\id_{2^n}\tr(A)$ \cite{huang2020predicting}. The random measurements from $\mc{E}_{\mathrm{Cl}}$ and $\mc{E}_{\mathrm{Pauli}}$ are denoted as Clifford and Pauli measurement primitives respectively. 

\begin{algorithm}[H]
\caption{Multi-shot shadow estimation}\label{algo:tShadow}
\begin{algorithmic}[1]
\Require
$M\times K$ sequentially prepared $\rho$
\Ensure
The shadow set $\{\hat{\rho}_{(i)}\}_{i=1}^M$. 
\For{$i= 1~\text{\textbf{to}}~M$} 
 \State Randomly choose $U\in\mc{E}$ and record it. 
 \For{$j= 1~\text{\textbf{to}}~K$} 
  \State Evolve the state $\rho$ using $U$ to get $U\rho U^{\dag}$.
  \State Measure the state in the computational basis $\{\ket{\mb{b}}\}$.
  \State Construct the unbiased estimator $\hat{\rho}_{(i)}^{(j)}$ with the result $\mb{b}^{(i,j)}$ by Eq.~\eqref{eq:shadowPost}, where $i$ and $j$ denoting the $j$-th shot under the $i$-th unitary.
  \EndFor
 \State Average $K$ results under the same unitary to get $\hat{\rho}_{(i)}=\frac1{K}\sum_j\hat{\rho}^{(j)}_{(i)}$.
\EndFor
\State Get the shadow set $\left\{\hat{\rho}_{(1)}, \widehat{\rho}_{(2)}, \cdots \hat{\rho}_{(M)}\right\}$, which contains $M$ independent estimators of $\rho$.
\end{algorithmic}
\end{algorithm}

In the multi-shot shadow estimation illustrated in Fig. \ref{fig:multi-shot}, one conducts $K$ shots by applying the same unitary $U$ sampled in each round, that is, the measurement settings of these $K$ shots are the same. So the total preparation-and-measurement number is $MK$. Denote the measurement result in $i$-th round and $j$-th shot as $\mb{b}^{(i,j)}$. One can construct the estimator of $\rho$ in this shot as $\hat{\rho}_{(i)}^{(j)}$ following Eq.~\eqref{eq:shadowPost}, and then average on $K$ shots and $M$ rounds to get
\begin{equation}\label{rhoij}
    \begin{aligned}
\hat{\rho}_{(i)}=K^{-1}\sum_{j\in[K]}\hat{\rho}^{(j)}_{(i)},\ \ \ \
\hat{\rho}=M^{-1}\sum_{i\in[M]} \hat{\rho}_{(i)}.
\end{aligned}
\end{equation}
This procedure is listed in Algorithm \ref{algo:tShadow}, and we remark that as $K=1$, it reduces to the original shadow estimation.

To estimate the expectation value of an observable $O$, one can construct the unbiased estimator as $\hat{O}=\tr(O\hat{\rho})$. The performance of shadow estimation is mainly characterized by the variance of $\hat{o}$, which determines the experiment time to control the estimation error,
\begin{equation}
    \begin{aligned}
     \mathrm{Var}\left(\hat{O}\right)&=\mathbb{E}\ \tr(O\widehat{\rho})^2- \tr(O\rho)^2.
    \end{aligned}
\end{equation}
Here, the expectation value is taken on the random unitary $U$ in $M$ rounds and the measurement result $\mb{b}^{(i,j)}$ in $MK$ shots. In the following sections, we give a general expression for this key quantity, and then analyze the variance for Pauli and Clifford measurements respectively.  We remark that even though we focus on estimating one observable $O$ here, the variance result can be applied to simultaneously estimating a set of observables $\{O^i\}$ by using the median-of-mean technique \cite{huang2020predicting}, which is discussed in Appendix \ref{App:GeneralO}.
\section{The general variance expression and cross-shadow-norm}
In this section, we give the general variance expression for the multi-shot shadow estimation and relate the variance to the cross-shadow-norm. First, we define two functions of a quantum state $\sigma$ and an observable $O$ as follows,
\begin{equation}\label{Gamma12}
    \begin{aligned}
&\Gamma_1(\sigma,O):=\mathbb{E}_U \sum_{\mb{b}} \bra{\mb{b}}U\sigma U^{\dag}\ket{\mb{b}} \bra{\mb{b}}U\mc{M}^{-1}(O)U^{\dag}\ket{\mb{b}}^2,\\
&\Gamma_2(\sigma,O):=\mathbb{E}_U  \sum_{\mb{b},\mb{b}'}\bra{\mb{b}}U\sigma U^{\dag}\ket{\mb{b}}\bra{\mb{b}}U\mc{M}^{-1}(O)U^{\dag}\ket{\mb{b}}  \bra{\mb{b}'}U\sigma U^{\dag}\ket{\mb{b}'}
       \bra{\mb{b}'}U\mc{M}^{-1}(O)U^{\dag}\ket{\mb{b}'},
\end{aligned}
\end{equation}
where $\mb{b}, \mb{b}'$ are computational basis for $n$-bit. The cross-shadow-norm is defined as follows. 
\begin{definition}
The cross-shadow-norm (short as XSnorm) of some observable $O$ is defined as 
\begin{equation}\label{eq:Xshadow}
\begin{aligned}
\|O \|_{\mathrm{Xshadow}}=\max_{\sigma\ \mathrm{state}}\  \Gamma_2(\sigma,O)^{\frac1{2}}.
    \end{aligned}
\end{equation}
\end{definition}
The proof of it being a norm is in Appendix  \ref{app:CSN}. Note that the original shadow norm (short as Snorm hereafter) is defined as $\|O\|_{\mathrm{shadow}}=\max_{\sigma} \Gamma_1(\sigma,O)^{\frac1{2}}$ \cite{huang2020predicting}. By definition, one has the upper bounds $\Gamma_1(\sigma,O)\leq \|O\|^2_{\mathrm{shadow}}$, and $\Gamma_2(\sigma,O)\leq \|O\|^2_{\mathrm{Xshadow}}$. Then we show the first result of this work considering the statistical variance for measuring some observable using multi-shot shadow estimation.
\begin{theorem}\label{Th:XSN}
    The statistical variance of $\tr(O \hat{\rho})$  shows
\begin{equation}\label{XSN1}
    \begin{aligned}
\mathrm{Var}\left[\tr(O  \hat{\rho})\right]&=\frac1{M}\left[\frac1{K}\Gamma_1(\rho,O_0)+(1-\frac1{K})\Gamma_2(\rho,O_0)-\tr(O_0\rho)^2\right],
\end{aligned}
\end{equation}
where $O_0=O-\tr(O)\id_D/D$ is the traceless part of $O$, and the functions $\Gamma_1,\Gamma_2$ are defined in Eq.~\eqref{Gamma12}, which can be bounded by the square of the Snorm and XSnorm respectively. As a result, one has the upper bound of the variance as
\begin{equation}\label{XSN2}
    \begin{aligned}
\mathrm{Var}\left[\tr(O  \hat{\rho})\right]\leq \frac1{M}\left[\frac1{K}\|O_0 \|^2_{\mathrm{shadow}}+(1-\frac1{K}) \|O_0 \|^2_{\mathrm{Xshadow}}\right].
\end{aligned}
\end{equation}
\end{theorem}
The proof of Theorem \ref{Th:XSN} is left in Appendix \ref{app:ThXSN}.
We remark that as $K=1$, the variance in Eq.~\eqref{XSN1} and the upper bound in Eq.~\eqref{XSN2} reduce to the result of the original shadow estimation \cite{huang2020predicting}. Note that the variance is an interpolation of the two functions $\Gamma_1(\rho,O_0)$ and $\Gamma_2(\rho,O_0)$ with the shot-number $K$. And the advantage of introducing multi-shot should come from the condition when $\Gamma_2(\rho,O_0)\ll\Gamma_1(\rho,O_0)$. We remark that Ref.~\cite{seif2023shadow} also investigated on this kind of multi-shot variance with a focus on quantum error mitigation.

The function $\Gamma_1(\sigma,O)$ and the Snorm $\|O_0 \|_{\mathrm{shadow}}$ have been extensively studied in Ref.~\cite{huang2020predicting}. Thus, the remaining content of this work is to analyze $\Gamma_2(\sigma,O)$ and the XSnorm $\|O_0 \|_{\mathrm{Xshadow}}$. To proceed, one can write $\Gamma_2(\sigma,O)$ in Eq.~\eqref{Gamma12} in the following twirling channel form.

\begin{equation}\label{eq:Gamma2new}
\begin{aligned}
\Gamma_2(\sigma,O)
=&\tr[\sigma\otimes \mc{M}^{-1}(O) \otimes \sigma\otimes \mc{M}^{-1}(O)\  \Phi^{(4,\mc{E})}\left( \Lambda_n \right)].
    \end{aligned}
\end{equation}
Here, we denote the 4-fold twirling channel as $\Phi^{(4,\mc{E})}(\cdot)=\mathbb{E}_{\{U\in \mc{E}\}} U^{\dag \otimes 4}(\cdot)U^{ \otimes 4}$, which is dependent on the random unitary ensemble $\mc{E}$, and the 4-copy $n$-qubit diagonal operator is defined as
\begin{equation}\label{eq:Lambdan}
\begin{aligned}
\Lambda_n:=\sum_{\mb{b},\mb{b}'} (\ket{\mb{b}}\bra{\mb{b}})^{\otimes 2}\otimes (\ket{\mb{b}'}\bra{\mb{b}'})^{\otimes 2}. 
    \end{aligned}
\end{equation}
For the global $n$-qubit Clifford ensemble, we denote the twirling channel as $\Phi^{(4,\mathrm{Cl})}_n$, and the one of local Clifford is just the tensor product $\bigotimes_{i=1}^n \Phi^{(4,\mathrm{Cl})}_{1,{(i)}}$.

It is known that Clifford ensemble is not a 4-design \cite{zhu2016clifford}. Hence, in order to calculate $\Gamma_2(\sigma,O)$ and the XSnorm, one should apply the representation theory (rep-th) of Clifford group in the 4-copy Hilbert space \cite{zhu2016clifford,gross2021schur}.
In the following sections, we first figure out the exact result of random Pauli measurements, say single-qubit Clifford twirling based on direct calculation, \emph{without} using the rep-th; and then give the bound of $\Gamma_2(\sigma,O)$ for random Clifford measurements, say $n$-qubit Clifford twirling with results from rep-th.

\section{Variance analysis for Pauli measurements}
In this section, we focus on the observable $O$ being an $n$-qubit Pauli operator $P$, that is $P=\otimes_{i=1}^nP_i$, with $P_i\in\{\id_2,X,Y,Z\}$. This choice is of practical interest, for example, in the measurement problem in quantum chemistry simulation \cite{Sam2020RMP,hadfield2022measurements,huang2021efficient}. We say a Pauli operator $P$ is $w$-weight, if there are $w$ qubits with $P_i\neq \id_2$, i.e., there are single-qubit Pauli operators on $w$ qubits and identity operators on the other $n-w$ qubits. We have the following result for the Pauli observable.
\begin{prop}\label{prop:pauliG2P}
    Suppose $O=P$ is a Pauli observable with weight $w$,  the function $\Gamma_2$ defined in Eq.~\eqref{Gamma12} shows
\begin{equation}\label{}
    \begin{aligned}
\Gamma_2^{\mathrm{Pauli}}(\sigma,P)=3^w \tr(P\sigma )^2
\end{aligned}
\end{equation}
for the random Pauli measurements. In this case, the XSnorm is $\|P\|^{\mathrm{Pauli}}_{\mathrm{Xshadow}}=3^{w/2}$ by Eq.~\eqref{eq:Xshadow}.
\end{prop}
We remark that compared to $\Gamma_1^{\mathrm{Pauli}}(\sigma,P)=3^w$ in the original shadow \cite{huang2020predicting}, $\Gamma_2^{\mathrm{Pauli}}(\sigma,P)$ has an extra relevant term, that is the square of the expectation value $\tr(P\sigma)^2$. To prove Proposition \ref{prop:pauliG2P}, we mainly apply the following Lemma \ref{lemma:SingleLambda} of the single-qubit Clifford twirling.

\begin{figure}[h]
\centering
\includegraphics[width=0.75\textwidth]{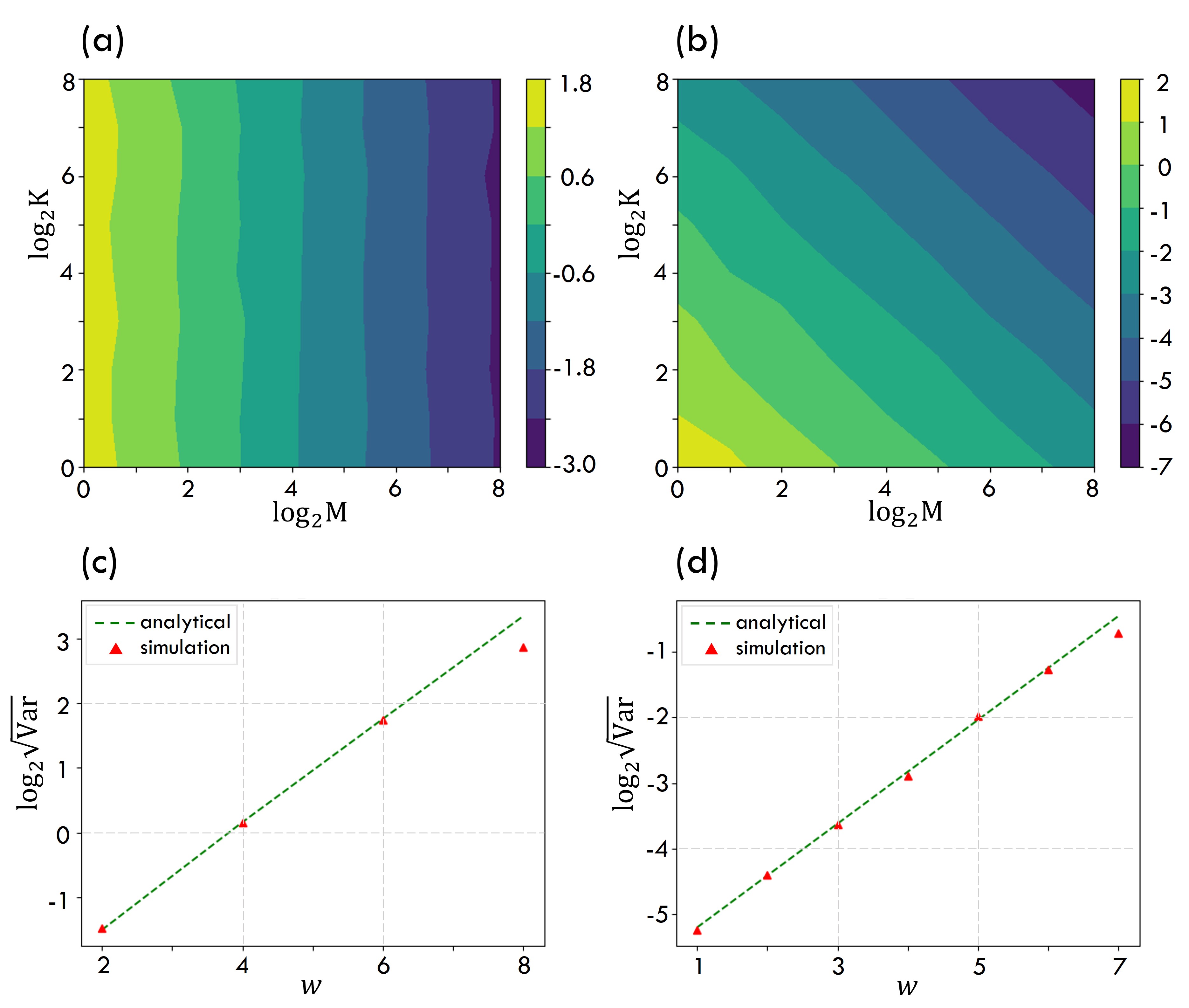}
\caption{Scaling of the statistical variance of estimating $\tr(P \widehat{\rho})$ with shot-number $K$, setting-number $M$ and weight $w$ for different Pauli observables. In (a) and (b), the processed state $\rho$ is a 5-qubit GHZ state. We investigate the dependence of $\log_{2}\sqrt{\mathrm{Var}\left[\tr(P \widehat{\rho})\right]}$ (colored area) on $K$ and $M$ with different Pauli observables $P=\sigma_z^1 \otimes \sigma_z^2 \otimes \id_2^{\otimes 3}$ and $P=\sigma_x^1 \otimes \sigma_x^2 \otimes \id_2^{\otimes 3}$ in (a) and (b) respectively.
The variance in (a) is almost independent with $K$, and for the one in (b) the dependence on $K$ and $M$ is the same. Note that in (a) and (b), $\tr(P\rho)^2=1$ and $0$ respectively, which are consistent with Eq.~\eqref{eq:pauli} and \eqref{eq:pauli-seperate}.
In (c) and (d), we explore the variance dependence on weight $w$. 
Here, the processed state $\rho$ is an 8-qubit GHZ state, and we set $M=K=64$. 
In (c), the Pauli observable is $P = \sigma_z^{\otimes w}\otimes \id_2^{\otimes(8-w)}$ with $w$ being an even number, and in (d), $P = \sigma_x^{\otimes w}\otimes \id_2^{\otimes(8-w)}$, such that $\tr(P\rho)^2=1$ and $0$ respectively. The red triangles represent the simulated variance which shows consistency with the green dotted line, i.e., the analytical variance given by Eq. \eqref{eq:pauli}.}
\label{fig:Pauli-numeric}
\end{figure}

\begin{lemma}\label{lemma:SingleLambda}
The $4$-fold single-qubit Clifford twirling channel maps the operator 
\begin{equation}\label{}
    \begin{aligned}
\Lambda_1:=\sum_{b=0,1}\sum_{b'=0,1}\ket{b}\bra{b}^{\otimes 2}\otimes \ket{b'}\bra{b'}^{\otimes 2}
    \end{aligned}
\end{equation} into
\begin{equation}\label{eq:SingleLambda}
    \begin{aligned}
\Phi^{(4,\mathrm{Cl})}_{1}(\Lambda_1)=\frac1{3}\left(\frac1{2}\id_2^{\otimes 2}\otimes\mbb{F}^{(2)}+\frac1{2}\mbb{F}^{(2)}\otimes \id_2^{\otimes 2}+\mbb{F}^{(4)}\right).
    \end{aligned}
\end{equation}
Here $\mbb{F}^{(t)}:=2^{-t/2}(\id_2^{\otimes t}+X^{\otimes t}+Y^{\otimes t}+Z^{\otimes t})$ on $t$ qubits, and $\mbb{F}^{(2)}$ is the swap operator on two qubits. 
\end{lemma}
We prove Lemma \ref{lemma:SingleLambda} by giving the general $t$-copy twirling result of $\Phi^{(t,\mathrm{Cl})}_{1}$ on any Pauli operator in Appendix \ref{app:lemmaSL}. 
Note that the diagonal operator in Eq.~\eqref{eq:Lambdan} can be decomposed to $\Lambda_{n}=\otimes_{i=1}^n\Lambda_{1,(i)}$, and the twirling result in Eq.~\eqref{eq:Gamma2new} is thus in the tensor-product form $\otimes_{i=1}^n\Phi^{(4,\mathrm{Cl})}_{1,{(i)}}(\Lambda_{1,(i)})$. Then one can prove Proposition \ref{prop:pauliG2P} by applying the result of Lemma \ref{lemma:SingleLambda} to Eq.~\eqref{eq:Gamma2new} and further calculation, which is left in Appendix \ref{app:propPauli}.

 Inserting $\Gamma_1^{\mathrm{Pauli}}(\sigma,P)=3^w$ \cite{huang2020predicting} and the result of $\Gamma_2^{\mathrm{Pauli}}(\sigma,P)$ of Proposition \ref{prop:pauliG2P} into Eq.~\eqref{XSN1}, one has the following exact variance.

\begin{theorem}\label{th:Pauli}
    Suppose $O=P$ is a Pauli observable with weight $w$, the variance to measure $P$ using random Pauli measurements is
\begin{equation}\label{eq:pauli}
    \begin{aligned}
\mathrm{Var}\left[\tr(P \widehat{\rho})\right]
=\frac1{M} \left[\frac1{K}3^w+(1-\frac1{K})3^w\tr(P \rho )^2-\tr( P \rho)^2\right].
    \end{aligned}
\end{equation}
\end{theorem}
Theorem \ref{th:Pauli} shows that the dependence of $\mathrm{Var}\left[\tr(P \widehat{\rho})\right]$ on the shot-number $K$ is related to $\tr(P \rho)^2$.
For the extreme cases when $\tr(\rho O)^2=0/1$, 
\begin{equation}\label{eq:pauli-seperate}
\mathrm{Var}\left[\tr(P  \widehat{\rho})\right]=\left\{\begin{aligned}
&\frac1{M}\frac1{K}3^w \ &\tr(P\rho)^2=0,\\
&\frac1{M} (3^w-1) \ &\tr(P\rho)^2=1.
    \end{aligned}\right.
\end{equation}
One can see that when $\tr(P\rho)^2=1$, the variance is independent of the shot-number $K$; for $\tr(P\rho)^2=0$, the dependence of $K$ is the same as the setting-number $M$. So it is advantageous to increase $K$ for $\tr(P\rho)^2$ being small, by considering that it is more convenient to repeat shots than change measurement settings in a real experiment.  The variance result here can be extended to any $w$-local observable following the proof routine in Ref.~\cite{huang2020predicting}, and this is investigated in Ref.~\cite{seif2023shadow}. We give more discussions on this point in Appendix \ref{App:GeneralO}.

In Fig. \ref{fig:Pauli-numeric}, we show numerical results of the variance dependence on $M$, $K$ and $w$, where we perform random Pauli measurements on the Greenberger–Horne–Zeilinger(GHZ) states. Please see Appendix \ref{AppendNum} for details of the numerical simulation, and also extended discussions about the reduction of experiment time-cost. The numerical results are consistent with the variance given by Eq.~\eqref{eq:pauli} and \eqref{eq:pauli-seperate}. 

\comments{
numerically study the dependence of $\mathrm{Var}\left[\tr(P \widehat{\rho})\right]$ on the shot-number $K$, the setting-number $M$ and the weight $w$. The simulations are implemented on quantum simulators, where we perform random Pauli measurements on GHZ states and use the measurement results to estimate the expectation value of Pauli observable $P$. As illustrated in Fig. \ref{fig:Pauli-numeric}, the statistical variance of $\tr(P \widehat{\rho})$}

\section{Variance analysis for Clifford measurements}
In this section, we analyze the statistical variance of random Clifford measurements. In this case, 
the inverse channel maps $\mc{M}^{-1}(O_0)=(D+1)O_0$, and $\Gamma_2(\sigma,O_0)$ in Eq.~\eqref{eq:Gamma2new} becomes 
\begin{equation}\label{eq:gammaClifford}
\begin{aligned}
\Gamma_2^{\mathrm{Cl}}(\sigma,O_0)=(D+1)^2\tr[\sigma\otimes O_0 \otimes \sigma\otimes O_0\  \Phi^{(4,\mathrm{Cl})}_n\left( \Lambda_n \right)].
    \end{aligned}
\end{equation}
For the simplicity of the calculation, we decompose $\Lambda_n=\Lambda_n^0+\Lambda_n^1$, with 
\begin{equation}\label{}
\begin{aligned}
&\Lambda_n^0=\sum_{\mb{b}} \ket{\mb{b}}\bra{\mb{b}}^{\otimes 4},\\
&\Lambda_n^1=\sum_{\mb{b}\neq \mb{b}'} \ket{\mb{b}}\bra{\mb{b}}^{\otimes 2}\otimes \ket{\mb{b}'}\bra{\mb{b}'}^{\otimes 2}.
    \end{aligned}
\end{equation}
To calculate $\Gamma_2^{\mathrm{Cl}}(\sigma,O_0)$, one should apply the 4-fold twirling result of the Clifford group \cite{zhu2016clifford} on $\Lambda_n^0$ and $\Lambda_n^1$, respectively. And we show in Appendix \ref{app:Cl} that $\Lambda_n^1$ contributes to the leading term of the final result. Since the $n$-qubit Clifford twirling is quite sophisticated, we do not calculate it in a very exact form but give the following upper bound on $\Gamma_2^{\mathrm{Cl}}(\sigma,O_0)$.

\begin{figure}[h]
\centering
\includegraphics[width=0.75\textwidth]{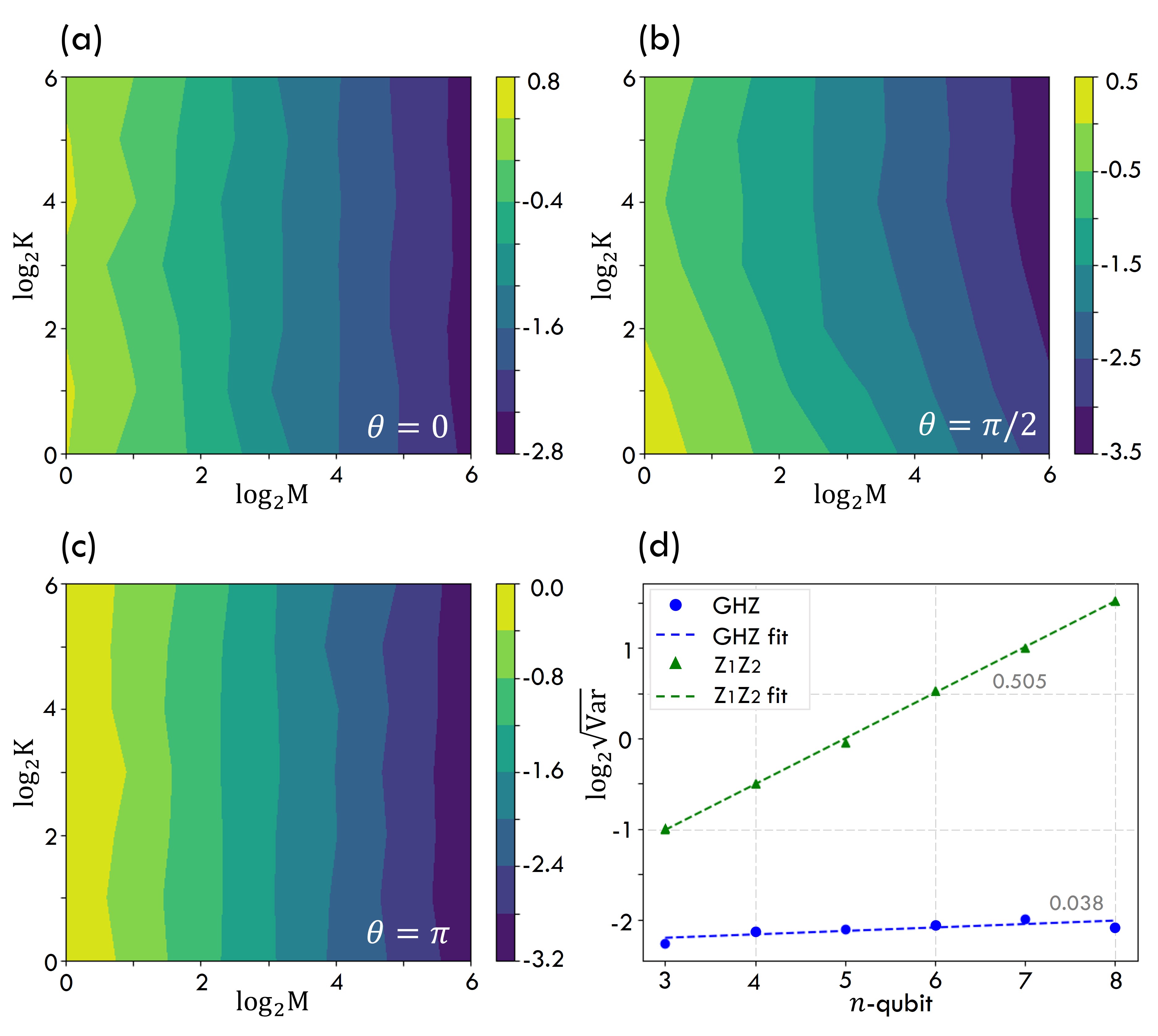}
\caption{Scaling of the statistical variance of random Clifford measurements with $K$, $M$ and qubit number $n$ for different processed state $\rho$ and observables. In (a), (b), and (c) the processed state $\rho$ is a 5-qubit $GHZ_\theta$ state, with $\theta$ equals to $0$, $\pi/2$ and $\pi$ respectively. and the observable is $O=\ket{GHZ_{\theta=0}}\bra{GHZ_{\theta=0}}$, with the expectation value $\tr(O\rho)=(1+\cos\theta)/2$ being $1$, $0.5$ and $0$ respectively. 
In all three cases, the variance is almost independent of the shot-number $K$, in contrast to Fig.~\ref{fig:Pauli-numeric} (b) of Pauli measurements.
In (d), we explore the variance dependence on qubit number $n$ for different observables when $M=K=32$. The blue dots represent the case when the observable $O=\ket{GHZ_{\theta=0}}\bra{GHZ_{\theta=0}}$, and the green triangles are the case when $O=\sigma_z^{\otimes 2}\otimes \id_2^{\otimes(n-2)}$.  Note that the corresponding Frobenius norm $\|O_0\|_2$ is about a constant and $2^{n/2}$, respectively. The dotted lines and the slopes represent the fitting curves for the corresponding cases, which is consistent with Eq.~\eqref{eq:Cl:var}.}
\label{fig:clifford-numeric}
\end{figure}

\begin{prop}\label{prop:Cl}
    Suppose $O_0$ is a traceless observable, the function $\Gamma_2$ defined in Eq.~\eqref{Gamma12} is upper bounded by
\begin{equation}\label{eq:prop:Cl}
    \begin{aligned}
\Gamma_2^{\mathrm{Cl}}(\sigma,O_0)\leq c \|O_0\|_2^2,
\end{aligned}
\end{equation}
for the random Clifford measurements,
where $\|A\|_2=\sqrt{\tr(AA^{\dag})}$ is the Frobenius norm, and $c$ is some constant independent of the dimension $D$. In this case, the XSnorm $\|O_0\|^{\mathrm{Cl}}_{\mathrm{Xshadow}}\leq \sqrt{c} \|O_0\|_2$ by Eq.~\eqref{eq:Xshadow}.
\end{prop}
For comparison, in the original shadow \cite{huang2020predicting},
\begin{equation}\label{}
    \begin{aligned}
\Gamma_1^{\mathrm{Cl}}(\sigma,O_0)&=\frac{D+1}{D+2}\left[\|O_0\|_2^2+2\tr(\sigma O_0^2)^2\right]\leq 3\|O_0\|_2^2,
\end{aligned}
\end{equation}
and $\|O_0\|_{\mathrm{shadow}}\leq \sqrt{3}\|O_0\|_2$, which is qualitatively same to the current multi-shot result.
The proof of Proposition \ref{prop:Cl} is left in Appendix \ref{app:Haar} and \ref{app:Cl}. In Appendix \ref{app:Haar}, we first analyze $\Gamma_2^{\mathrm{Haar}}(\sigma,O_0)$ in the Haar random case, where the global unitary is sampled from Haar ensemble or any unitary 4-design ensemble. We then extend the result to the Clifford ensemble in Appendix \ref{app:Cl}. For both Haar and Clifford ensembles, we can bound $\Gamma_2(\sigma,O_0)$ in the order $\mc{O}(1)\|O_0\|_2^2$ as in Proposition \ref{prop:Cl}.

Inserting Eq.~\eqref{eq:prop:Cl} in Eq.~\eqref{XSN1}, one has the following variance result.
\begin{theorem}
    For a general observable $O$, the variance to measure $O$ using the random Clifford measurements is upper bounded by
\begin{equation}\label{eq:Cl:var}
    \begin{aligned}
\mathrm{Var}\left[\tr(O \widehat{\rho})\right]
&\leq\frac1{M} \left[\frac{3}{K}+(1-\frac1{K})c\right]\|O_0\|_2^2\\
&\leq \frac{c'}{M}\|O_0\|_2^2,
    \end{aligned}
\end{equation}
with $c'=c+3$ some constant independent of the dimension $D$, and $\|A\|_2=\sqrt{\tr(AA^{\dag})}$ is the Frobenius norm.
\end{theorem}
In the original shadow, the variance bound is about $\mc{O}(1)M^{-1}\|O_0\|_2^2$. Even though Eq.~\eqref{eq:Cl:var} is an upper bound, it already gives a hint that it is not very helpful to increase the shot-number $K$ in the random Clifford measurements. We demonstrate this phenomenon with numerical simulation considering measuring the fidelity. In Fig.~\ref{fig:clifford-numeric}, we take the state $\rho=\ket{GHZ_{\theta}}\bra{GHZ_{\theta}}$, with $\ket{GHZ_{\theta}}=1/\sqrt{2}(\ket{0}^{\otimes n}+e^{i\theta}\ket{1}^{\otimes n})$, i.e., the GHZ state with some phase. And the observable is taken to be $O=\ket{GHZ_{\theta=0}}\bra{GHZ_{\theta=0}}$, that is, measuring the fidelity of the quantum state to the GHZ state. We find that the variance is almost independent with the shot-number $K$ for various choices of $\theta$ which gives different expectation values of $\tr(O\rho)$. This is very different from the results of random Pauli measurements as shown in Fig.~\ref{fig:Pauli-numeric}.

\section{Conclusion and outlook}
In this work, we systemically analyze the statistical performance of multi-shot shadow estimation by introducing a key quantity \textemdash the cross-shadow-norm. We find the advantage of this framework in Pauli measurements, however, the advantage in Clifford measurements is not significant. There are a few interesting points that merit further investigation. First, the application of the result of Pauli measurements to quantum chemistry problems is promising, especially combined with other measurement techniques, such as derandomization and local-bias \cite{huang2021efficient,hadfield2022measurements}. Second, considering that the advantage of multi-shot for Clifford measurements is minor, it is interesting to investigate whether this hold or not for shallow circuit \cite{hu2023classical,bertoni2022shallow,arienzo2022closed} or restricted unitary evolution \cite{Hu2022Hamiltonian,tran2023measuring,mcginley2022shadow,van2022hardware}. Third, it is also intriguing to extend the current analysis to boson \cite{Gandhari2022CV,Becker2022cv,gu2022efficient} and fermion systems \cite{Zhao2021Fermionic,Low2022fermion}, and to nonlinear functions of quantum states \cite{elben2020mixedstate,zhou2022hybrid,liu2022detecting}. Finally, we expect the result here would also benefit the statistical analysis for other randomized measurement tasks, especially those related to high-order functions \cite{singlezhou,Andreas2019Multipartite,Yu2021Moments,Zhenhuan2022correlation}.

While completing this work, we became aware of a related work \cite{Helsen2022Thrifty} by Jonas Helsen and Michael Walter which also considers shadow estimation under multi-shot.
Different from Ref.~\cite{Helsen2022Thrifty}, we analyze the performance under Pauli measurements. Our negative conclusion on Clifford measurements is consistent with theirs. Ref. \cite{Helsen2022Thrifty} also demonstrates the weakness of Clifford measurements by giving the exact variance for estimating stabilizer-fidelity. Additionally, they nicely extend the Clifford unitary to doped Clifford circuits to bypass this weakness.
 
\section{Acknowledgements}
Y.Z. thanks Huangjun Zhu for the useful discussion on unitary design.
This work is supported by NSFC Grant No.~12205048 and startup funding from Fudan University. 





\begin{appendix}

\section{Proofs of cross-shadow-norm}
\subsection{Proof of cross-shadow-norm as a norm}\label{app:CSN}
First one has $\|0\|_{\mathrm{Xshadow}}=0$, and then we verify the triangle inequality as follows.
\begin{equation}
    \begin{aligned}
     &\|O^1+O^2\|_{\mathrm{Xshadow}}^2\\
    =& \mathbb{E}_U \sum_{\mb{b},\mb{b}'}\bra{\mb{b}}U\sigma U^{\dag}\ket{\mb{b}}\bra{\mb{b}}U[\mc{M}^{-1}(O^1)+\mc{M}^{-1}(O^2)]U^{\dag}\ket{\mb{b}}\\
    &\ \ \ \ \ \ \ \ \ \times \bra{\mb{b}'}U\sigma U^{\dag}\ket{\mb{b}'}
    \bra{\mb{b}'}U[\mc{M}^{-1}(O^1)+\mc{M}^{-1}(O^2)]U^{\dag}\ket{\mb{b}'}\\
       =& \Gamma_2(\sigma,O^1)+\Gamma_2(\sigma,O^2)
       \\&+2\mathbb{E}_U \sum_{\mb{b},\mb{b}'}\bra{\mb{b}}U\sigma U^{\dag}\ket{\mb{b}}\bra{\mb{b}}U\mc{M}^{-1}(O^1)U^{\dag}\ket{\mb{b}}\bra{\mb{b}'}U\sigma U^{\dag}\ket{\mb{b}'}
       \bra{\mb{b}'}U\mc{M}^{-1}(O^2)U^{\dag}\ket{\mb{b}'}\\
       \leq &\|O^1\|_{\mathrm{Xshadow}}^2+\|O^2\|_{\mathrm{Xshadow}}^2\\
       &+ 2\mathbb{E}_U \sum_{\mb{b}}\bra{\mb{b}}U\sigma U^{\dag}\ket{\mb{b}}\bra{\mb{b}}U\mc{M}^{-1}(O^1)U^{\dag}\ket{\mb{b}}\sum_{\mb{b}'}\bra{\mb{b}'}U\sigma U^{\dag}\ket{\mb{b}'}
       \bra{\mb{b}'}U\mc{M}^{-1}(O^2)U^{\dag}\ket{\mb{b}'}\\
       \leq &\|O^1\|_{\mathrm{Xshadow}}^2+\|O^2\|_{\mathrm{Xshadow}}^2+2\sqrt{\Gamma_2(\sigma,O^1)}\sqrt{\Gamma_2(\sigma,O^2)}\\
       \leq &\|O^1\|_{\mathrm{Xshadow}}^2+\|O^2\|_{\mathrm{Xshadow}}^2+2\sqrt{\|O^1\|_{\mathrm{Xshadow}}\|O^2\|_{\mathrm{Xshadow}}}\\
       =&(\|O^1\|_{\mathrm{Xshadow}}+\|O^2\|_{\mathrm{Xshadow}})^2.
    \end{aligned}
\end{equation}
Here the first and third inequalities are due to the definition of XSnorm in Eq.~\eqref{eq:Xshadow}, and $\sigma$ is assumed to be the optimal state in the maximization for $O^1+O^2$, but may be not the optimal one for $O^1$ and $O^2$ respectively. The second inequality is by the Cauchy–Schwarz inequality for the domain of the random unitary $U$. That is, we can take the formula of the cross term as the inner product of two vectors indexed by $U$.

\subsection{Proof of Theorem \ref{Th:XSN}}\label{app:ThXSN}
As $\widehat{\rho}_{(i)}$ are $M$ i.i.d. random variables, we only need to focus on a specific $i_0$, and get the final variance by directly dividing it by $M$.
By definition the variance of  $\tr(O  \widehat{\rho}_{(i_0)})$ shows
\begin{equation}
    \begin{aligned}
     \mathrm{Var}\left[\tr(O  \widehat{\rho}_{(i_0)})\right]&=\mathbb{E}\ \tr(O\widehat{\rho}_{(i_0)})^2- \tr(O\rho)^2\\
     &=\mathbb{E}\ \tr(O_0\widehat{\rho}_{(i_0)})^2- \tr(O_0\rho)^2.
    \end{aligned}
\end{equation}
Here one can shift the operator to its traceless part without changing the variance. The expectation value can be written explicitly as
\begin{equation}\label{}
    \begin{aligned}
       \mathbb{E}\ \tr(O\widehat{\rho}_{(i_0)})^2&=\mathbb{E}\ \left[\frac1{K}\sum_j\tr(O\widehat{\rho}^{(j)}_{(i_0)})\right]^2\\
       &=\frac1{K^2}\sum_{j,j'}\ \mathbb{E}\  \tr(O\widehat{\rho}^{(j)}_{(i_0)})\tr(O\widehat{\rho}^{(j')}_{(i_0)}).
    \end{aligned}
\end{equation}

The expectation value of the terms in the summation depends on the coincidence of the index $j$.
For $j=j'$ with totally $K$ terms, one has 
\begin{equation}\label{}
    \begin{aligned}
    &\mathbb{E}\  \tr(O\widehat{\rho}^{(j)}_{(i_0)})\tr(O\widehat{\rho}^{(j)}_{(i_0)})\\
       =&\mathbb{E}\  \bra{\hat{\mb{b}}^{(j)}}U\mc{M}^{-1}(O)U^{\dag}\ket{\hat{\mb{b}}^{(j)}}^2  \\
       =&\mathbb{E}_U\sum_{\mb{b}^{(j)}}\mathrm{Pr}(\mb{b}^{(j)}|U)\ \bra{\mb{b}^{(j)}}U\mc{M}^{-1}(O)U^{\dag}\ket{\mb{b}^{(j)}}
       ^2\\
       =& \mathbb{E}_U \sum_{\mb{b}^{(j)}} \bra{\mb{b}^{(j)}}U\rho U^{\dag}\ket{\mb{b}^{(j)}} \bra{\mb{b}^{(j)}}U\mc{M}^{-1}(O)U^{\dag}\ket{\mb{b}^{(j)}}
       ^2=\Gamma_1(\rho,O).
    \end{aligned}
\end{equation}
by definition in Eq.~\eqref{Gamma12}.
Here in the first equality we insert the definition of $\widehat{\rho}^{(j)}_{(i_0)}$, and use the self-adjoint property of the inverse channel $\mc{M}^{-1}$.

For $ j\neq j'$ with totally $K^2-K$ terms, the measurements are under the same setting $i_0$ but for different shots, one has 
\begin{equation}\label{}
    \begin{aligned}
    &\mathbb{E}\  \tr(O\widehat{\rho}^{(j)}_{(i_0)})\tr(O\widehat{\rho}^{(j')}_{(i_0)})\\
       =&\mathbb{E}\  \bra{\mb{b}^{(j)}}U\mc{M}^{-1}(O)U^{\dag}\ket{\hat{\mb{b}}^{(j)}}  \bra{\mb{b}^{(j')}}U\mc{M}^{-1}(O)U^{\dag}\ket{\hat{\mb{b}}^{(j')}}\\
       =&\mathbb{E}_U \sum_{\mb{b}^{(j)},\mb{b}^{(j')}}\mathrm{Pr}(\mb{b}^{(j)}|U)\ \mathrm{Pr}(\mb{b}^{(j')}|U) \ \bra{\mb{b}^{(j)}}U\mc{M}^{-1}(O)U^{\dag}\ket{\hat{\mb{b}}^{(j)}}
       \bra{\mb{b}^{(j')}}U\mc{M}^{-1}(O)U^{\dag}\ket{\hat{\mb{b}}^{(j')}}\\
       =& \mathbb{E}_U \sum_{\mb{b}^{(j)},\mb{b}^{(j')}} \bra{\mb{b}^{(j)}}U\rho U^{\dag}\ket{\mb{b}^{(j)}}\bra{\mb{b}^{(j')}}U\rho U^{\dag}\ket{\mb{b}^{(j')}}\ \bra{\mb{b}^{(j)}}U\mc{M}^{-1}(O)U^{\dag}\ket{\mb{b}^{(j)}}
       \bra{\mb{b}^{(j')}}U\mc{M}^{-1}(O)U^{\dag}\ket{\mb{b}^{(j')}}\\
       =& \Gamma_2(\rho,O),
    \end{aligned}
\end{equation}
by definition in Eq.~\eqref{Gamma12}, and here the subscript $i_0$ is omitted without ambiguity.

As a result, the total variance shows
\begin{equation}
    \begin{aligned}
     \mathrm{Var}\left[\tr(O  \widehat{\rho})\right]&=\frac1{M}\mathrm{Var}\left[\tr(O  \widehat{\rho}_{(i_0)})\right]=\frac1{M}\mathrm{Var}\left[\tr(O_0  \widehat{\rho}_{(i_0)})\right]\\
     &=\frac1{M}\left[\frac1{K}\Gamma_1(\rho,O_0)+(1-\frac1{K})\Gamma_2(\rho,O_0)-\tr(O_0\rho)^2\right],
    \end{aligned}
\end{equation}
and we finish the proof.

\section{Variance analysis of a collection of observables}\label{App:GeneralO}
In this section, we extend the variance results on a single observable $O$ in the main text to a collection of observable $\{O^i\}_{i=1}^L$. In particular, we are interested in the summation of them $H=\sum_i O^i$, and each $O^i=\alpha_i P_i$ could be proportional to a Pauli operator $P_i$. For instance, $H$ can be a Hamiltonian of an $n$-qubit system, and we aim to estimate the expectation value of $H$ within error $\epsilon$. We first give general formulas of the measurement budget or the variance without a specific chosen random measurements, and then show the result of estimating $H$ under random Pauli measurements, which is of practical interest. For simplicity, we denote the estimator based on the shadow $\hat{\rho}$ as $\hat{O}=\tr(O\hat{\rho})$ and the expectation value as $\bar{O}=\tr(O\rho)$ in the following discussion.

We mainly give two solutions to this task. The first one is directly using the median-of-mean technique for all $\{O^i\}_{i=1}^L$ following Theorem S1 of Ref.~\cite{huang2020predicting}. Remembering that we totally collect $M$ i.i.d. estimators $\{\hat{\rho}_{(k)}\}$, and each of them is composed of measurement data from repeated $K$-shots as shown in Eq.~\eqref{rhoij}. For the median-of-mean post-processing, $\{\hat{\rho}_{(k)}\}$ is grouped into $M_1$ sets with each set containing $M_2$ elements (please refer to Eq.~(S12) in Ref.~\cite{huang2020predicting} for more details). The estimator for each $O^i$ is denoted as $\hat{O^i}(M_1,M_2)$, and $\hat{H}(M_1,M_2)$ is thus the summation of them. Different from estimating each $O^i$ within error $\epsilon$ there, we should instead keep the whole error within $\epsilon$. Let $\epsilon_i=\epsilon\|O^i\|_{\infty}/\sum_i\|O^i\|_{\infty}$ and $\epsilon=\sum_i \epsilon_i$, one can find the following result.
\begin{corollary}\label{Co:mofm}
    Set $M_1=2\log(2L/\delta)$ and $M_2=34\epsilon^{-2}\left(\sum_i\|O^i\|_{\infty}\right)^2\max_{j\in[L]} \mathrm{Var}(\hat{O^j}/\|O^j\|_{\infty})$, with the variance dependent on the particular random measurement primitive. A collection of $M=M_1M_2$ shadow estimators $\{\hat{\rho}_{(k)}\}_{k=1}^M$ via the the median-of-mean prediction is sufficient to make the estimating error of $H=\sum_i O^i$, $|\hat{H}(M_1,M_2)-\bar{H}|\leq \epsilon$, under confidence level $1-\delta$.
    
    In particular, for each $O^i=\alpha_i P_i$ and the random Pauli measurements, the number of independent estimators is at most
\begin{equation}\label{eq:Acol}
    \begin{aligned}
M 
=68\log(2L/\delta)\epsilon^{-2}\|\vec{\alpha}\|^2_1\max_{j\in[L]} \mathrm{Var}(\hat{P_j}),
\end{aligned}
\end{equation}
Here $\|\vec{\alpha}\|_1=\sum_i|\alpha_i|$ is the 1-norm of the coefficient vector, and $\mathrm{Var}(\hat{P_j})$ of Pauli operator $P_j$ on a single estimator $\hat{\rho}_{(k_0)}$ given in Eq.~\eqref{eq:pauli} (i.e., $M=1$ there).
\end{corollary}
Here the error of estimating $H$ is bounded by the triangle inequality of each error $\epsilon_i$ of estimating $O^i$.

The second way is to calculate and bound the variance of $\hat{H}$, based on the variance of each $\hat{O^i}$. Let us first define some functions of a state $\sigma$ and two observable $O^1$ and $O^2$ generalized from Eq.~\eqref{Gamma12}.
\begin{equation}\label{Gamma12app}
    \begin{aligned}
&\Gamma_1(\sigma,O^1,O^2):=\mathbb{E}_U \sum_{\mb{b}} \bra{\mb{b}}U\sigma U^{\dag}\ket{\mb{b}} \bra{\mb{b}}U\mc{M}^{-1}(O^1)U^{\dag}\ket{\mb{b}}\bra{\mb{b}}U\mc{M}^{-1}(O^2)U^{\dag}\ket{\mb{b}},\\
&\Gamma_2(\sigma,O^1,O^2):=\mathbb{E}_U  \sum_{\mb{b},\mb{b}'}\bra{\mb{b}}U\sigma U^{\dag}\ket{\mb{b}}\bra{\mb{b}}U\mc{M}^{-1}(O^1)U^{\dag}\ket{\mb{b}}  \bra{\mb{b}'}U\sigma U^{\dag}\ket{\mb{b}'}
       \bra{\mb{b}'}U\mc{M}^{-1}(O^2)U^{\dag}\ket{\mb{b}'}.
\end{aligned}
\end{equation}
One has the following variance result of estimating $H$.
\begin{theorem}\label{Th:VarGapp}
   By using multi-shot shadow estimation with $M$-round and $K$-shot, the statistical variance of $\hat{H}$ with $H=\sum_i O^i$ shows
\begin{equation}\label{VarGapp1}
    \begin{aligned}
\mathrm{Var}(\hat{H})&=\sum_i \mathrm{Var}(\hat{O^i})+2\sum_{i<j} \mathrm{Cov} (\hat{O^i},\hat{O^j}).
\end{aligned}
\end{equation}
The variance of $\mathrm{Var}(\hat{O^i})$ is given in Eq.~\eqref{XSN1}, and the covariance can be calculated via
\begin{equation}\label{coVarGapp}
    \begin{aligned}
\mathrm{Cov} (\hat{O^i},\hat{O^j})=\frac1{M}\left[\frac1{K}\Gamma_1(\rho,O^i,O^j)+(1-\frac1{K})\Gamma_2(\rho,O^i,O^j)-\bar{O^i}\bar{O^j}\right].
\end{aligned}
\end{equation}
\end{theorem}
The proof follows similarly to that of Theorem \ref{Th:XSN} in Sec.~\ref{app:ThXSN}. In principle, one can calculate $\mathrm{Cov} (\hat{O^i},\hat{O^j})$ via Eq.~\eqref{coVarGapp}, similar as we have done in main text for $\mathrm{Var}(\hat{O^i})$. For the random Pauli measurements, Ref.~\cite{huang2020predicting} and Ref.~\cite{seif2023shadow} investigate $\Gamma_1(\rho,O^i,O^j)$ and $\Gamma_2(\rho,O^i,O^j)$ for two Pauli observables, respectively. 

Here for a supplement, we further give some bound for the total variance $\mathrm{Var}(\hat{H})$. Note that $\mathrm{Cov} (\hat{O^i},\hat{O^j})\leq \sqrt{\mathrm{Var}(\hat{O^i})\mathrm{Var}(\hat{O^j})}$ by Cauchy–Schwarz inequality, and thus one has 
$\mathrm{Var}(\hat{H})\leq \left[\sum_i \sqrt{\mathrm{Var}(\hat{O^i})} \right]^2$. Applying this to random Pauli measurements, one finally gets the following result.

\begin{corollary}\label{}
    For $H=\sum_i \alpha_i P_i$ and the random Pauli measurements, the variance of estimating $H$ is upper bounded by 
\begin{equation}\label{VarHPauli}
    \begin{aligned}
\mathrm{Var}(\hat{H})\leq \left[\sum_i \sqrt{\mathrm{Var}(\alpha_i\hat{P_i})} \right]^2\leq \|\vec{\alpha}\|^2_1\max_{j\in[L]} \mathrm{Var}(\hat{P_j}).
\end{aligned}
\end{equation}
with $\mathrm{Var}(\hat{P_j})$ given in Eq.~\eqref{eq:pauli}.
\end{corollary}
The second inequality in Eq.~\eqref{VarHPauli} is by the Holder's inequality ($p=1,q=\infty$). 

By adopting the median-of-mean technique only for a single observerble $H$, one sets
$M_1=2\log(2/\delta)$ and $M_2=34\epsilon^{-2}\mathrm{Var}(\hat{H})$, with the variance on a single estimator $\hat{\rho}_{(k_0)}$. In this way, a collection shadow estimators $\{\hat{\rho}_{(k)}\}_{i=1}^M$ is sufficient to make the estimator error less than $\epsilon$ under confidence $1-\delta$. By using upper bound in Eq.~\eqref{VarHPauli}, one has $M$ at most being

\begin{equation}\label{app:Mbetter}
    \begin{aligned}
M&=68\log(2/\delta)\epsilon^{-2} \left[\sum_i \sqrt{\mathrm{Var}(\alpha_i\hat{P_i})} \right]^2\\
&\leq 68\log(2/\delta)\epsilon^{-2} \|\vec{\alpha}\|^2_1\max_{j\in[L]} \mathrm{Var}(\hat{P_j}).
\end{aligned}
\end{equation}
It is clear that the upper bound of $M$ obtained here is better than the one in Eq.~\eqref{eq:Acol}. 

We remark that one can also apply Cauchy–Schwarz inequality to Eq.~\eqref{VarHPauli} to find another upper bound $\mathrm{Var}(\hat{H})\leq \|\vec{\alpha}\|^2_2\sum_j \mathrm{Var}(\hat{P_j})$. Note that $\mathrm{Var}(\hat{P_j})$ is a function of $w_i$, $\bar{P_i}$ and the shot-number $K$. One may further utilizes the commuting relation among the operators $\{P_i\}$ of interest to minimize the measurement budget, and we leave it to some further study.

\section{Proof for random Pauli measurements}\label{}
\subsection{Proof of Lemma \ref{lemma:SingleLambda}}\label{app:lemmaSL}
We first decompose $\Lambda_1$ defined in Lemma \ref{lemma:SingleLambda} into the Pauli operator representation as follows.
\begin{equation}\label{}
\begin{aligned}
\Lambda_1&= \sum_{b}\ket{b}\bra{b}^{\otimes 2}\sum_{b'}\ket{b'}\bra{b'}^{\otimes 2}\\
&=\frac1{2}(\id_2\otimes \id_2+Z\otimes Z)\otimes \frac1{2}(\id_2\otimes \id_2+Z\otimes Z)\\
 &=\frac1{4}(\id_2^{\otimes 4}+\id_2^{\otimes 2}Z^{\otimes 2}+Z^{\otimes 2}\id_2^{\otimes 2}+Z^{\otimes 4})
    \end{aligned}
\end{equation}
Here $\id_2$ and $Z$ are single-qubit identity and Pauli $Z$ operator. By inserting this decomposition inside the 4-copy single-qubit Clifford twirling channel, 
\begin{equation}\label{}
\begin{aligned}
 &\Phi^{(4,\mathrm{Cl})}_1[\frac1{4}(\id_2^{\otimes 4}+\id_2^{\otimes 2}Z^{\otimes 2}+Z^{\otimes 2}\id_2^{\otimes 2}+Z^{\otimes 4})]\\
 =&\frac1{4}[\id_2^{\otimes 4}+\id_2^{\otimes 2}\otimes \Phi_1^{(2,\mathrm{cl})}(Z^{\otimes 2})+\Phi_1^{(2,\mathrm{cl})}(Z^{\otimes 2})\otimes \id_2^{\otimes 2}+\Phi_1^{(4,\mathrm{cl})}(Z^{\otimes 4})]\\
 =& \frac1{4}[\id_2^{\otimes 4}+\id_2^{\otimes 2}\frac1{3}(2\mbb{F}^{(2)}-\id^{\otimes 2})+\frac1{3}(2\mbb{F}^{(2)}-\id^{\otimes 2})\id_2^{\otimes 2}+(4\mbb{F}^{(4)}-\id_2^{\otimes 4})/3]\\
=&1/6\id_2^{\otimes 2}\mbb{F}^{(2)}+1/6\mbb{F}^{(2)}\id_2^{\otimes 2}+1/3\mbb{F}^{(4)}.
    \end{aligned}
\end{equation}
In the second line: the first term is due to $\Phi_1^{(4,\mathrm{Cl})}$ being unital, the two terms in the middle reduce to 2-copy twirling, and the last term is still a 4-copy twirling. In the third line, we insert the result of the following Lemma.

\begin{lemma}\label{}
The $t$-fold single-qubit Clifford twirling channel maps any Pauli operator $P_1^{\otimes t}$ with $P_1 \in \{X,Y,Z\}$ into
\begin{equation}\label{}
    \Phi_1^{(t,\mathrm{Cl})}(P_1^{\otimes t})=\left\{\begin{aligned}
&\frac1{3}(X^{\otimes t}+Y^{\otimes t}+Z^{\otimes t})=(2^{t/2}\mbb{F}^{(t)}-\id^{\otimes t}_2)/3, \ \ &t\ \mathrm{even}\\
&0, \ \ &t\ \mathrm{odd}
    \end{aligned}\right.
\end{equation}
where in the first line we define $\mbb{F}^{(t)}:=2^{-t/2}(\id_2^{\otimes t}+X^{\otimes t}+Y^{\otimes t}+Z^{\otimes t})$, and $\mbb{F}^{(2)}$ is the swap operator on two qubits.
\end{lemma}
  \begin{proof}
By definition, random single-qubit Clifford gate takes $P_1 \in \{X,Y,Z\}$ uniformly to all six directions in the Bloch sphere, i.e., $\{\pm X,\pm Y,\pm Z\}$. Consequently, the final result is the average of six terms, i.e., $1/6\sum_{P_1\in \{\pm X,\pm Y,\pm Z\}}P_1^{\otimes t}$. In the even $t$ case, we thus have the equal weight summation of $X^{\otimes t},Y^{\otimes t},Z^{\otimes t}$. In the odd $k$ case, $X^{\otimes t}$ and $(-X)^{\otimes t}$ cancel with each other, same for $Y,Z$ terms, thus it returns zero.
\end{proof}

\subsection{Proof of Proposition \ref{prop:pauliG2P}}\label{app:propPauli}

\begin{proof}
Without loss of generality, suppose $P=P_1\otimes \cdots P_w \otimes \id_{[n-w]}$ with the first $w$-qubit owning Pauli operator, and of course traceless.  
The inverse channel maps $\mc{M}^{-1}(O)=\bigotimes_{i=1}^n \mc{M}_1^{-1}(O)=\bigotimes_{i=1}^w 3P_i \otimes \id_{[n-w]}$.
By definition of Eq.~\eqref{Gamma12}, one has

\begin{equation}\label{}
\begin{aligned}
&\mathbb{E}_U \sum_{\mb{b},\mb{b}'}\bra{\mb{b}}U\sigma U^{\dag}\ket{\mb{b}}\bra{\mb{b}'}U\sigma U^{\dag}\ket{\mb{b}'}\ \bra{\mb{b}}U\mc{M}^{-1}(O)U^{\dag}\ket{\mb{b}}
       \bra{\mb{b}'}U\mc{M}^{-1}(O)U^{\dag}\ket{\mb{b}'}\\
       =&\mathbb{E}_{U=\bigotimes_{i=1}^n u_i} \sum_{b_1,\cdots,b_n;b_1',\cdots,b_n'}\tr[\sigma^{\otimes 2}\ \bigotimes_{i=1}^n (u_i^{\dag}\ket{b_i}\bra{b_i}u_i)\bigotimes_{i=1}^n (u_i^{\dag}\ket{b_i'}\bra{b_i'}u_i)] \\
       &\ \ \ \ \ \ \ \ \ \ \ \ \ \ \ \ \ \ \ \ \ \ \ \ \ \ \ \ \ \ \ \ \times \prod_{i=1}^w\bra{b_i}u_i3P_i u_i^{\dag}\ket{b_i}\prod_{i=1}^w\bra{b_i'}u_i3P_i u_i^{\dag}\ket{b_i'}
      \\=&\mathbb{E}_{\bigotimes_{i=1}^w  u_i} \sum_{b_1,\cdots,b_w;b_1',\cdots,b_w'} \tr[\sigma^{\otimes 2}\ \bigotimes_{i=1}^w  (u_i^{\dag}\ket{b_i}\bra{b_i}u_i) \otimes \id_{[n-w]} \bigotimes_{i=1}^w  (u_i^{\dag}\ket{b_i'}\bra{b_i'}u_i)\otimes \id_{[n-w]}] \\
      &\ \ \ \ \ \ \ \ \ \ \ \ \ \ \ \ \ \ \ \ \ \ \ \ \ \ \ \ \ \ \ \ \times\prod_{i=1}^w\bra{b_i}u_i3P_i  u_i^{\dag}\ket{b_i}\prod_{i=1}^w\bra{b_i'}u_i3P_i u_i^{\dag}\ket{b_i'}\\
      \\=&\mathbb{E}_{\bigotimes_{i=1}^w  u_i} \sum_{b_1,\cdots,b_w;b_1',\cdots,b_w'} \tr[\sigma_{[m]}^{\otimes 2}\ \bigotimes_{i=1}^w  (u_i^{\dag}\ket{b_i}\bra{b_i}u_i) \bigotimes_{i=1}^w  (u_i^{\dag}\ket{b_i'}\bra{b_i'}u_i)] \\
      &\ \ \ \ \ \ \ \ \ \ \ \ \ \ \ \ \ \ \ \ \ \ \ \ \ \ \ \ \ \ \ \ \times\prod_{i=1}^w\bra{b_i}u_i3P_i  u_i^{\dag}\ket{b_i}\prod_{i=1}^w\bra{b_i'}u_i3P_i u_i^{\dag}\ket{b_i'}\\
=& 3^{2w}\tr[(\sigma_{[w]}\bigotimes_{i=1}^w P_i)^{\otimes 2}\  \ \mathbb{E}_{\bigotimes_{i=1}^w  u_i} \sum_{b_1,\cdots,b_w;b_1',\cdots,b_w'}\bigotimes_{i=1}^w[(u_i^{\dag}\ket{b_i}\bra{b_i}u_i)^{\otimes 2} \otimes (u_i^{\dag}\ket{b_i'}\bra{b_i'}u_i)^{\otimes 2}] ]\\
=&3^{2w}\tr[(\sigma_{[w]}\bigotimes_{i=1}^w P_i)^{\otimes 2}\ \bigotimes_{i=1}^w \Phi^{(4,\mathrm{Cl})}_{1,{(i)}}\left(\sum_{b_i}\ket{b_i}\bra{b_i}^{\otimes 2}\sum_{b_i'}\ket{b_i'}\bra{b_i'}^{\otimes 2}\right)]\\
    \end{aligned}
\end{equation}
Here in the first equality, we write the random unitary $U=\bigotimes_{i=1}^n u_i$ and computational basis summation $\mb{b}=\{b_1b_2\cdots b_n\},\mb{b}'=\{b'_1b'_2\cdots b'_n\}$ on the single-qubit level; in the second equality we sum the $b_i,b_i'$ for $i>w$ in the trace, and it gives $\id_{[n-w]}$ on the last $n-w$ qubits no matter what $u_i$ is chosen. And the problem is reduced on the first $w$-qubit in the third equality, with $\sigma_{[m]}$ the reduced state of the first $m$-qubit from $\sigma$. In the last two lines, we arrange these operators and relate the result to the single-qubit Clifford twirling.

By inserting the Eq.~\eqref{eq:SingleLambda} of Lemma.~\ref{lemma:SingleLambda}, one finally has
\begin{equation}\label{}
\begin{aligned}
\Gamma_2(\sigma,P)
=& 3^{w}\tr[(\sigma_{[w]}\bigotimes_{i=1}^w P_i)^{\otimes 2} \bigotimes_{i=1}^w\left(\frac1{2}\id_i^{\otimes 2}\otimes \mbb{F}_i^{(2)}+\frac1{2}\mbb{F}_i^{(2)}\otimes \id_i^{\otimes 2}+\mbb{F}_i^{(4)}\right)]\\
=& 3^{w}\tr[(\sigma_{[w]}\bigotimes_{i=1}^w P_i)^{\otimes 2} \bigotimes_{i=1}^w\left(\mbb{F}_i^{(4)}\right)]\\
=& 3^w \tr(\sigma_{[w]}\bigotimes_{i=1}^w P_i)^2=3^w \tr(\sigma P)^2.
    \end{aligned}
\end{equation}
Here the second line is by the following fact. Suppose there is  an appearance of $\id_i^{\otimes 2}\otimes \mbb{F}_i^{(2)}$ or $\mbb{F}_i^{(2)}\otimes \id_i^{\otimes 2}$ on $i$-th qubit, then the identity operator on the first or last two-copy would result in $\tr(P_i)=0$, such that gives no contribution to the final result. Consequently, we are left only $\mbb{F}_i^{(4)}$. In the final line, $\bigotimes_{i=1}^w\left(\mbb{F}_i^{(4)}\right)$ contains all possible $w$-qubit Pauli operators, and only the one $\bigotimes_{i=1}^wP_i$ contributes to the final result.
\end{proof}

\section{Statistical analysis for Haar random measurements}\label{app:Haar}
In this section, we aim to give the statistical analysis for the global Haar random measurement, that is, the unitary ensemble $\mc{E}$ is Haar random on $\mc{H}_D$ (or any unitary 4-design ensemble).  The central result shows as follows, as an analog of Proposition \ref{prop:Cl} in main text.
\begin{prop}\label{prop:Haar}
    Suppose $O_0$ is a traceless observable, the function $\Gamma_2$ defined in Eq.~\eqref{Gamma12} is upper bounded by
\begin{equation}\label{}
    \begin{aligned}
\Gamma_2^{\mathrm{Haar}}(\sigma,O_0)\leq c_1 \|O_0\|_2^2,
\end{aligned}
\end{equation}
for the random Haar measurements, where $\|A\|_2=\sqrt{\tr(AA^{\dag})}$ is the Frobenius norm, and $c_1$ is some constant independent of the dimension $D$.  In this case, the XSnorm $\|O_0\|_{\mathrm{Xshadow}}^{\mathrm{Haar}}\leq \sqrt{c_1} \|O_0\|_2$ by Eq.~\eqref{eq:Xshadow}.
\end{prop}
The main task of this section is to prove this proposition, which is helpful for the discussion of random Clifford measurements shown in the next section. 

Recall the essential quantity in Eq.~\eqref{eq:gammaClifford} we would like to calculate shows 
\begin{equation}\label{eq:HaarAll}
\begin{aligned}
\Gamma_2^{\mathrm{Haar}}&=(D+1)^2\tr[\sigma\otimes O_0 \otimes \sigma\otimes O_0\  \Phi^{(4,\mathrm{Haar})}_n\left( \Lambda_n \right)]\\
&=(D+1)^2\tr[\sigma\otimes O_0 \otimes \sigma\otimes O_0\  \Phi^{(4,\mathrm{Haar})}_n\left( \Lambda_n^0+\Lambda_n^1 \right)]
    \end{aligned}
\end{equation}
with the only difference being that we use the Haar twirling $\Phi^{(4,\mathrm{Haar})}$ here, and $\Lambda_n$ is decomposed into two parts as in main text, 
\begin{equation*}\label{}
\begin{aligned}
\Lambda_n^0&=\sum_{\mb{b}} \ket{\mb{b}}\bra{\mb{b}}^{\otimes 4},\\
\Lambda_n^1&=\sum_{\mb{b}\neq \mb{b}'} \ket{\mb{b}}\bra{\mb{b}}^{\otimes 2}\otimes \ket{\mb{b}'}\bra{\mb{b}'}^{\otimes 2}.
    \end{aligned}
\end{equation*}

Before we calculate their contributions separately in the next subsections, we give a brief review of the result of the $t$-copy Haar twirling, and more details can be found in, for example Ref.~\cite{Collins2006Integration,Roberts2017Chaos,Roth2018Recovering}. Denote the permutation elements of the $t$-th order symmetric group as $\pi\in S_t$, and there is a unitary representation of  $\pi$ on the $t$-copy Hilbert space $\mc{H}_D^{\otimes t}$ as
\begin{equation}\label{}
    \begin{aligned}
T_{\pi}\ket{\mb{b}_1,\mb{b}_2,\cdots,\mb{b}_t}=\ket{\mb{b}_{\pi(1)},\mb{b}_{\pi(2)},\cdots,\mb{b}_{\pi(t)}},
    \end{aligned}
\end{equation}
with $\ket{\mb{b}}$ the basis state for one copy. By the celebrated Schur–Weyl duality, the twirling result is related to the irreducible representation (ir-rep) of $S_t$ as follows.
\begin{lemma}\label{}
The $t$-fold Haar twirling channel maps $A\in \mc{H}_D^{\otimes t}$ into
\begin{equation}\label{eq:HaarTt}
    \begin{aligned}
&\Phi^{(t,\mathrm{Haar})}(A)= \frac1{t!}\sum_{\pi\in S_t}\tr(A T_\pi)T_{\pi^{-1}}\sum_{\lambda} \frac{d_\lambda}{D_\lambda}P_\lambda,\\
    \end{aligned}
\end{equation}
where $\lambda$ denotes the irreducible representation of $S_t$, and $P_{\lambda}$ is the corresponding projector showing
\begin{equation}\label{eq:PLambda}
    \begin{aligned}
P_\lambda=\frac{d_{\lambda}}{t!}\sum_{\pi\in S_t}\chi^{\lambda}(\pi)T_\pi
    \end{aligned},
\end{equation}
with $\chi^{\lambda}(\pi)$ being the character of $\pi$.

\comments{
$Q$ is a stabilizer code subspace projector, commuting with $P_{\lambda}$,
\begin{equation}\label{}
    \begin{aligned}
Q=\frac1{D^2}\sum_k W_k^{\otimes 4}
    \end{aligned}
\end{equation}
with $W_k$ being all $D^2$ $n$-qubit Pauli operators.
}

\end{lemma}

\subsection{The contribution of $\Lambda^0_n$ in Haar case}
Inserting the term $\Lambda^0_n$ in the 4-copy Haar twirling channel in Eq.~\eqref{eq:HaarTt} with $t=4$ one has

\begin{equation}\label{}
    \begin{aligned}
\Phi^{(4,\mathrm{Haar})}(\Lambda^0_n)&= \frac1{4!}\sum_{\pi\in S_4}\sum_{\mb{b}}\tr(\ket{\mb{b}}\bra{\mb{b}}^{\otimes 4} T_{\pi^{-1}})T_\pi\sum_{\lambda} \frac{d_\lambda}{D_\lambda}P_\lambda\\
&=\frac{D}{4!}\sum_{\pi\in S_4}T_\pi\sum_{\lambda} \frac{d_\lambda}{D_\lambda}P_\lambda\\
&=DP_{\mathrm{sym}}\sum_{\lambda} \frac{d_\lambda}{D_\lambda}P_\lambda\\&
=\frac{D}{D_{\sym}}P_{\sym}=\frac{D}{4!D_{\sym}}\sum_{\pi\in S_k}\pi
\end{aligned}
\end{equation}
Here the second line is by the fact $\tr(\ket{\mb{b}}\bra{\mb{b}}^{\otimes 4} T_{\pi^{-1}})=1$ for any $\mb{b}$ and $\pi$; the third line is by the definition of the symmetric subspace as $P_{\sym}=\frac{1}{4!}\sum_{\pi\in S_4}T_\pi$; the final line is by the fact that symmetric subspace is one of the ir-rep subspace and orthogonal to others.
Inserting this twirling result and the dimension $D_{\sym}=(D+3)(D+2)(D+1)D/4!$, one has
\begin{equation}\label{eq:HaarL0}
\begin{aligned}
&(D+1)^2\tr[\sigma\otimes O_0 \otimes \sigma\otimes O_0\  \Phi^{(4,\mathrm{Haar})}\left(\Lambda^0_n\right)]\\
=&\frac{(D+1)^2D}{(D+3)(D+2)(D+1)D} \sum_{\pi\in S_k}\tr(\sigma\otimes O_0 \otimes \sigma\otimes O_0\ T_\pi)\\
=&\mathcal{O}(D^{-1})\sum_{\pi\in S_4}\tr(\sigma\otimes O_0 \otimes \sigma\otimes O_0\ T_\pi)= \mathcal{O}( D^{-1}) \tr(O_0^2)
    \end{aligned}
\end{equation}
The final line shows that the result is in the order $\mathcal{O}(D^{-1})\tr(O_0^2)$, by the following lemma.
\begin{lemma}\label{prop:piO2}
For any $\pi\in S_4$, the following inequality holds for a quantum state $\sigma$ and a traceless observable $O_0$,
\begin{equation}\label{}
    \begin{aligned}
\tr(\sigma\otimes O_0 \otimes \sigma\otimes O_0\ T_\pi)\leq \tr(O_0^2).
    \end{aligned}
\end{equation}
\end{lemma}
\begin{proof}
For the permuation operator $T_\pi$, the formula could take the following values: $\tr(\sigma^2)\tr(O_0^2)$, $\tr(O_0^2)$, $\tr(O_0^2\sigma)$, $\tr(O_0^2\sigma^2)$, $\tr(O_0\sigma O_0\sigma)$. All these can be bounded by $\tr(O_0^2)$. We bound the last one with Cauchy–Schwarz inequality for operator as follows.
\begin{equation}\label{}
\begin{aligned}
\tr(O_0\sigma O_0\sigma)\leq \|O_0\sigma \|_2^2=\sqrt{\tr(O_0\sigma \sigma O_0)}^2=\tr(O_0^2\sigma^2)\leq \|O_0^2\|_{\infty}\leq \tr(O_0^2).
    \end{aligned}
\end{equation}
\end{proof}

\subsection{The contribution of $\Lambda_n^1$ in Haar case}
For the second term of $\Lambda_n^1$, by inserting the Haar twirling channel in Eq.~\eqref{eq:HaarTt} with $t=4$, one has 
\begin{equation}\label{}
    \begin{aligned}
\Phi^{(4,\mathrm{Haar})}(\Lambda_n^1)&= \frac1{4!}\sum_{\pi\in S_4}\sum_{\mb{b}\neq \mb{b}'}\tr(\ket{\mb{b}}\bra{\mb{b}}^{\otimes 2}\otimes \ket{\mb{b}'}\bra{\mb{b}'}^{\otimes 2}T_{\pi^{-1}})T_\pi\sum_{\lambda} \frac{d_\lambda}{D_\lambda}P_\lambda\\
&=\frac1{4!}(D^2-D) \sum_{\pi\in S_r}\pi\sum_{\lambda} \frac{d_\lambda}{D_\lambda}P_\lambda\\
&=\frac1{6}(D^2-D) \sum_{\lambda} \frac{d_\lambda}{D_\lambda}P_rP_\lambda.
\end{aligned}
\end{equation}
Here the second line is by the fact that the trace formula gives zero unless $\pi\in S_r$, with the set $S_r=\{(),(12),(34),(12)(34)\}$ being a subgroup of $S_4$. And we define the corresponding projector as $P_r=\frac1{4}\sum_{\pi\in S_r}T_\pi$. One thus further gets
\begin{equation}\label{eq:lambda1}
\begin{aligned}
(D+1)^2\tr[\sigma\otimes O_0 \otimes \sigma\otimes O_0\ \Phi^{(4,\mathrm{Haar})}\left(\Lambda_n^1\right)]=&\frac1{6}(D^2-D)(D+1)^2 \sum_{\lambda} \frac{d_\lambda}{D_\lambda}\tr(\sigma\otimes O \otimes \sigma\otimes O\ P_rP_\lambda)
\end{aligned}
\end{equation}

Recall the definition of  $P_{\lambda}$ in Eq.~\eqref{eq:PLambda}, and for each ir-rep $\lambda$, one has
\begin{equation}\label{eq:Haaruselatter}
\begin{aligned}
\frac{d_\lambda}{D_\lambda}\tr(\sigma\otimes O_0 \otimes \sigma\otimes O_0\ P_rP_\lambda)&=\frac{d_\lambda^2}{D_\lambda}\frac1{4*4!}\sum_{\pi'\in S_r,\pi\in S_4}\chi^{\lambda}(\pi)\tr(\sigma\otimes O_0 \otimes \sigma\otimes O_0\ T_{\pi'}T_\pi )\\
&=\frac{d_\lambda^2}{D_\lambda}\frac1{4*4!}\sum_{\pi'\in S_r,\pi\in S_4}\chi^{\lambda}(\pi)\tr(\sigma\otimes O_0 \otimes \sigma\otimes O_0\ T_{\pi'\pi} )\\
&\leq \frac{d_\lambda^2}{D_\lambda}|\chi^{\lambda}(\pi)|_{\max \pi\in S_4}\tr(O_0^2)
    \end{aligned}
\end{equation}
where the last line we use Lemma \ref{prop:piO2} for each $T_{\pi'\pi}$. Here $|\chi^{\lambda}(\pi)|_{\max \pi \in S_4}\leq 3$ , $\frac{d_\lambda^2}{D_\lambda}=\mathcal{O}(D^{-4})$ for any  ir-rep $\lambda$, and there are totally $5$ ir-reps \cite{zhu2016clifford}. As a result one has

\begin{equation}\label{eq:HaarL1}
\begin{aligned}
(D+1)^2\tr[\sigma\otimes O_0 \otimes \sigma\otimes O_0\ \Phi^{(4,\mathrm{Haar})}\left(\Lambda^1_n\right)]&<\frac1{6}(D^2-D)(D+1)^2 \left[\sum_{\lambda} \frac{d_\lambda^2}{D_\lambda}|\chi^{\lambda}(\pi)|_{\max \pi\in S_4}\right]\ \tr(O_0^2)\\
&=\mathcal{O}(1)\tr(O_0^2)
    \end{aligned}
\end{equation}
Inserting Eq.~\eqref{eq:HaarL0} and Eq.~\eqref{eq:HaarL1} into Eq.~\eqref{eq:HaarAll}, we finish the proof of Proposition \ref{prop:Haar}.

\section{Statistical analysis for  random Clifford measurements}\label{app:Cl}
Similar as the Haar case, here the we should calculate the quantity
in Eq.~\eqref{eq:gammaClifford} shown as follows,
\begin{equation}\label{eq:Cliffordall}
\begin{aligned}
\Gamma_2^{\mathrm{Cl}}(\sigma,O_0)=(D+1)^2\tr[\sigma\otimes O_0 \otimes \sigma\otimes O_0\  \Phi_n^{(4,\mathrm{Cl})}\left( \Lambda_n^0+\Lambda_n^1 \right)].
    \end{aligned}
\end{equation}
In the following subsections, as in the Haar case, we evaluate the contributions from $\Lambda_n^0$ and $\Lambda_n^1$ separately. As Clifford circuit is not a 4-design, the 4-copy twirling result is a little different compared to Eq.~\eqref{eq:HaarTt}, which is shown as follows. One can refer to Ref.~\cite{zhu2016clifford,Roth2018Recovering,Leone2021quantumchaosis} for more details.
\begin{lemma}\label{}
The $4$-fold $n$-qubit Clifford twirling channel maps $A\in \mc{H}_D^{\otimes 4}$ into
\begin{equation}\label{eq:CliffordTt}
    \begin{aligned}
\Phi_n^{(4,\mathrm{Cl})}(A)= \frac1{4!}\sum_{\lambda} d_\lambda\sum_{\pi\in S_4}[\frac1{D_\lambda^+}\tr(AQT_\pi)T_{\pi^{-1}}Q+\frac1{D_\lambda^-}\tr(AQ^{\perp}T_\pi)T_{\pi^{-1}} Q^{\perp}]P_\lambda
    \end{aligned}
\end{equation}
where $\lambda$ denotes the irreducible representation of $S_t$, and $P_{\lambda}$ is the corresponding projector shown
in Eq.~\eqref{eq:PLambda} with $t=4$. The operator
$Q$ is a stabilizer code subspace projector, commuting with any $T_\pi$ and $P_{\lambda}$,
\begin{equation}\label{eq:Q}
    \begin{aligned}
Q=\frac1{D^2}\sum_k W_k^{\otimes 4}
    \end{aligned}
\end{equation}
with $W_k$ running on all $D^2$ $n$-qubit Pauli operators including the identity; and $Q^{\perp}=\id_D^{\otimes 4}-Q$.
\end{lemma}

\subsection{The contribution of $\Lambda^0_n$ in Clifford case}
For the first term $\Lambda^0_n$, by using the twirling formula in Eq.~\eqref{eq:CliffordTt} one has
\begin{equation}\label{eq:L0Cliff}
    \begin{aligned}
\Phi_n^{(4,\mathrm{Cl})}(\Lambda^0_n)&= \frac1{4!}\sum_{\lambda} d_\lambda\sum_{\pi\in S_4}\sum_{\mb{b}}[\frac1{D_\lambda^+}\tr(\ket{\mb{b}}\bra{\mb{b}}^{\otimes 4} QT_{\pi^{-1}})T_\pi Q+\frac1{D_\lambda^-}\tr(\ket{\mb{b}}\bra{\mb{b}}^{\otimes 4}Q^{\perp}T_{\pi^{-1}})T_\pi Q^{\perp}]P_\lambda\\
&=\frac1{4!}\sum_{\lambda} d_\lambda\sum_{\pi\in S_4}\sum_{\mb{b}}[\frac1{D_\lambda^+}\tr(\ket{\mb{b}}\bra{\mb{b}}^{\otimes 4} Q)T_\pi Q+\frac1{D_\lambda^-}\tr(\ket{\mb{b}}\bra{\mb{b}}^{\otimes 4}Q^{\perp})T_\pi Q^{\perp}]P_\lambda\\
&=\frac1{4!}\sum_{\lambda} d_\lambda\sum_{\pi\in S_4}D[\frac1{D_\lambda^+D}Q +\frac1{D_\lambda^-}(1-1/D)Q^{\perp}]T_\pi P_\lambda\\
&=\sum_{\lambda} d_\lambda [\frac1{D_\lambda^+}Q +\frac{D-1}{D_\lambda^-}Q^{\perp}]\sum_{\pi\in S_4}\frac1{4!} T_\pi\ P_\lambda\\
&=[\frac1{D_{\sym}^+}Q +\frac{D-1}{D_{\sym}^-}Q^{\perp}]P_{\sym}.
\end{aligned}
\end{equation}
Here in the second line, we use the fact that $T_\pi\ket{\mb{b}}^{\otimes4}=\ket{\mb{b}}^{\otimes4}$ for any $\pi$. In the third line, $\tr(\ket{\mb{b}}\bra{\mb{b}}^{\otimes 4} Q)=1/D$, since $W_k$ which only contains $\{\id_2,Z\}$ on each qubit contributes to it and there are totally $2^n=D$ such terms. Therefore,  $\tr(\ket{\mb{b}}\bra{\mb{b}}^{\otimes 4} Q^{\perp})=1-1/D$. In the final line, we only have $P_{\sym}$ due to the orthogonality of each $P_{\lambda}$.
Inserting this twirling result into the trace formula in Eq.~\eqref{eq:Cliffordall} one gets

\begin{equation}\label{eq:ClL0}
\begin{aligned}
&(D+1)^2\tr[\sigma\otimes O_0 \otimes \sigma\otimes O_0\  \Phi_n^{(4,\mathrm{Cl})}\left(\Lambda^0_n\right)]\\
=&(D+1)^2(\frac1{D_{\sym}^+}-\frac{D-1}{D_{\sym}^-})\tr[\sigma\otimes O_0 \otimes \sigma\otimes O_0\ QP_{\sym}]+(D+1)^2\frac{D-1}{D_{\sym}^-}\tr[\sigma\otimes O_0 \otimes \sigma\otimes O_0\ P_{\sym}]\\
=&\mathcal{O}(1)\tr[\sigma\otimes O_0 \otimes \sigma\otimes O_0\ QP_{\sym}]+ \mathcal{O}(D^{-1}) \tr[\sigma\otimes O_0 \otimes \sigma\otimes O_0\ P_{\sym}]\\
=&\mathcal{O}(D^{-1})\tr(O_0^2)
    \end{aligned}
\end{equation}
where in the third line we use the fact 
$D_{\sym}^+= \mathcal{O}(D^{-2})$ and $D_{\sym}^-=\mathcal{O}(D^{-4})$ \cite{zhu2016clifford}; in the last line, the second term is directly from the result in Eq.~\eqref{eq:HaarL0} of the Haar case, and the first term is on account of the following Lemma \ref{prop:piO2Qnew}, by taking $P_0=P_{\sym}$ and $P_1=\id$ there. In total, 
we thus show that the contribution from $\Lambda^0_n$ is still in the order $\mathcal{O}(D^{-1})\tr(O^2)$ as in the Haar case.
\begin{lemma}\label{prop:piO2Qnew}
For any two projectors $P_0$ and $P_1$, which may not commute with each other but commute with $Q$,  the following inequality holds for a quantum state $\sigma$ and an observable $O$,
\begin{equation}\label{}
    \begin{aligned}
\tr(\sigma\otimes O \otimes \sigma\otimes O\ Q P_0P_1)\leq D^{-1}\tr(O^2).
    \end{aligned}
\end{equation}
\end{lemma}
\begin{proof}
Suppose the spectrum decomposition of the observable is $O=\sum_j a_j \ket{\Psi_j}\bra{\Psi_j}$, and we define the operator $O'=\sum_j i^{\delta(a_j<0)}|a_j|^{\frac1{2}}\ket{\Psi_j}\bra{\Psi_j}$, such that $O'^2=O$ and $O'O'^{\dag}=\sqrt{O^2}$.
\begin{equation}\label{eq:QP0P1}
    \begin{aligned}
\tr(\sigma\otimes O \otimes \sigma\otimes O \ QP_0P_1)&=\tr(\sigma^{\frac1{2}}\otimes O' \otimes \sigma^{\frac1{2}} \otimes O'Q\ P_0P_1 Q\sigma^{\frac1{2}}\otimes O' \otimes \sigma^{\frac1{2}} \otimes O')\\
&\leq \|\sigma^{\frac1{2}}\otimes O' \otimes \sigma^{\frac1{2}} \otimes O' Q\|_2\ \ \|P_0P_1 Q\sigma^{\frac1{2}}\otimes O' \otimes \sigma^{\frac1{2}} \otimes O'\|_2\\
&\leq \|\sigma^{\frac1{2}}\otimes O' \otimes \sigma^{\frac1{2}} \otimes O' Q\|_2^2=\tr(\sigma \otimes O'O'^{\dag} \otimes \sigma \otimes O'O'^{\dag} Q)\\
    \end{aligned}
\end{equation}
the first inequality is by using the Cauchy–Schwarz inequality, and the second inequality is by 
\begin{equation}\label{}
    \begin{aligned}
\|P_0P_1 Q \sigma^{\frac1{2}}\otimes O' \otimes \sigma^{\frac1{2}} \otimes O'\|_2&=\sqrt{\tr(\sigma \otimes O'O'^{\dag} \otimes \sigma \otimes O'O'^{\dag}QP_1P_0P_1Q)}\\
&\leq \sqrt{\tr(\sigma \otimes O'O'^{\dag} \otimes \sigma \otimes O'O'^{\dag}Q)}=\|\sigma^{\frac1{2}}\otimes O' \otimes \sigma^{\frac1{2}} \otimes O' Q\|_2,
    \end{aligned}
\end{equation}
since $P_1P_0P_1\leq P_1\leq \id$. 

To further bound Eq.~\eqref{eq:QP0P1}, we insert the formula of $Q$ in Eq.~\eqref{eq:Q} to get
\begin{equation}\label{}
    \begin{aligned}
    \tr(\sigma\otimes O \otimes \sigma\otimes O \ QP_0P_1)&\leq\tr(\sigma \otimes O'O'^{\dag} \otimes \sigma \otimes O'O'^{\dag} Q)\\
    &= D^{-2}\sum_k\tr(\sigma W_k)^2 \tr(\sqrt{O^2}W_k)^2 \\
    &\leq D^{-2}\sum_k\tr(\sqrt{O^2}W_k)^2\\
    &= D^{-1} \tr(O^2).
    \end{aligned}
\end{equation}
Here the second inequality is by the fact $\tr(\sigma W_k)^2\leq 1$, and the last line is by the decomposition of the $n$-qubit swap operator in the Pauli basis as $\sw=D^{-1}\sum_k W_k \otimes W_k$.
\end{proof}

\subsection{The contribution of $\Lambda^1_n$ in Clifford case}
For the second term $\Lambda_n^1$, by using the twirling formula in Eq.~\eqref{eq:CliffordTt} one has
\begin{equation}\label{eq:phi4L1}
    \begin{aligned}
\Phi_n^{(4,\mathrm{Cl})}(\Lambda_n^1)= \frac1{4!}\sum_{\lambda} d_\lambda \sum_{\pi\in S_4} \sum_{\mb{b}\neq \mb{b}'}&\bigg{[}\frac1{D_\lambda^+}\tr(\ket{\mb{b}}\bra{\mb{b}}^{\otimes 2}\otimes \ket{\mb{b}'}\bra{\mb{b}'}^{\otimes 2} QT_\pi)T_{\pi^{-1}}Q\\
&+\frac1{D_\lambda^-}\tr(\ket{\mb{b}}\bra{\mb{b}}^{\otimes 2}\otimes \ket{\mb{b}'}\bra{\mb{b}'}^{\otimes 2} Q^{\perp}T_\pi)T_{\pi^{-1}} Q^{\perp}\bigg{]}P_\lambda
    \end{aligned}
\end{equation}

We define the first trace formula as
\begin{equation}\label{}
    \begin{aligned}
Q_0(\pi):=& \tr(\ket{\mb{b}}\bra{\mb{b}}^{\otimes 2}\otimes \ket{\mb{b}'}\bra{\mb{b}'}^{\otimes 2} QT_\pi)\\
=& D^{-2} \sum_k \bra{\mb{b}}W_k\otimes \bra{\mb{b}}W_k\otimes  \bra{\mb{b}'}W_k\otimes \bra{\mb{b}'}W_k\  T_\pi \ \ket{\mb{b}}\otimes\ket{\mb{b}}\otimes \ket{\mb{b}'}\otimes\ket{\mb{b}'}.
    \end{aligned}
\end{equation}
There are two cases which give nonzero value depending on $\pi$.
In the first case, via the permutation $T_\pi$, $\ket{\mb{b}}$ is connected to $\bra{\mb{b}}$, and this also holds for $\mb{b}'$. That is, $\pi\in S_r$ with the set $S_r=\{(),(12),(34),(12)(34)\}$ and one has $|\bra{\mb{b}}W_k\ket{\mb{b}}|^2|\bra{\mb{b}'}W_k\ket{\mb{b}'}|^2$. As in the Haar case, one can only choose $W_k$ with $\{\id_2,Z\}$ for each qubit, so it gives $D/D^2=D^{-1}$. In the second case, all $\ket{\mb{b}}$ are connected to $\bra{\mb{b}'}$, that is $\pi\in S_{r'}$, with $S_{r'}=(13)(24)S_r=S_r(13)(24)$. And one has $|\bra{\mb{b}}W_k\ket{\mb{b}'}|^4$.
For $\mb{b}\neq \mb{b}'$, one needs to choose $\{\id_2,Z\}$ for the qubit with the same bit value of $\mb{b}$ and $\mb{b}'$, and $\{X,Y\}$ for the qubit of different bit values. Thus the result is also $D^{-1}$.
Besides these two cases, the term shows $|\bra{\mb{b}}W_k\ket{\mb{b}'}|^2\bra{\mb{b}}W_k\ket{\mb{b}}|\bra{\mb{b}'}W_k\ket{\mb{b}'}$. It is not hard to see that it is $0$ no matter what $W_k$ is.

Define $Q_1(\pi):=\tr(\ket{\mb{b}}\bra{\mb{b}}^{\otimes 2}\otimes \ket{\mb{b}'}\bra{\mb{b}'}^{\otimes 2} Q^{\perp}T_\pi)$. Note that $Q_0(\pi)+Q_1(\pi)=1$ when $\pi \in S_r$, otherwise it is $0$, as shown in the Haar case. We thus have
\begin{equation}\label{}
   Q_0(\pi)=\left\{ \begin{aligned}
     &D^{-1},&\pi\in S_r\cup S_{r'}\\
         &0,&\pi\in S_4\setminus(S_r\cup S_{r'})  
   \end{aligned}\right.\ \ \ \ \ 
   Q_1(\pi)=\left\{ \begin{aligned}
     &1-D^{-1},&\pi\in S_r\\
     &-D^{-1},&\pi\in S_{r'}\\
    &0,&\pi\in S_4\setminus(S_r\cup S_{r'})  
   \end{aligned}\right.
\end{equation}

Inserting these into Eq.~\eqref{eq:phi4L1}, and define the projector $P_{r'}=\frac1{4}\sum_{\pi\in S_{r'}}T_\pi$ similar as $P_r$ in the Haar case, one has
\begin{equation}\label{}
    \begin{aligned}
&\Phi_n^{(4,\mathrm{Cl})}(\Lambda_n^1)\\
=&\frac1{4!}\sum_{\lambda} d_\lambda (D^2-D) \left\{\frac1{D_\lambda^+} D^{-1}\sum_{\pi\in S_r\cup S_{r'} }T_\pi Q+ \frac1{D_\lambda^-} \left[(1-D^{-1}) \sum_{\pi\in S_r}T_\pi -D^{-1}\sum_{\pi'\in S_{r'}}T_{\pi'}\right] Q^{\perp}\right\}P_\lambda\\
=&\sum_{\lambda} \Theta_1(\lambda) QP_{r}P_{\lambda}+\Theta_2(\lambda) QP_{r'}P_{\lambda}+\Theta_3(\lambda) P_{r}P_{\lambda}+\Theta_4(\lambda) P_{r'}P_{\lambda})
    \end{aligned}
\end{equation}
with
\begin{equation}\label{}
    \begin{aligned}
&\Theta_1(\lambda)=\frac1{3!}d_\lambda (D^2-D)(\frac1{D_\lambda^+D}-\frac{1-D^{-1}}{D_\lambda^-})=\mathcal{O}(D^{-1}),\\ &\Theta_2(\lambda)=\frac1{3!}d_\lambda (D^2-D)(\frac1{D_\lambda^+D}+\frac1{D_\lambda^-D})=\mathcal{O}(D^{-1}),\\
&\Theta_3(\lambda)=\frac1{3!}d_\lambda (D^2-D)\frac{1-D^{-1}}{D_\lambda^-}=\mathcal{O}(D^{-2}),\\
&\Theta_4(\lambda)=-\frac1{3!}d_\lambda (D^2-D)\frac1{D_\lambda^-D}=\mathcal{O}(D^{-3}),\\
    \end{aligned}
\end{equation}
by the fact $D_\lambda^+=\mathcal{O}(D^2)$ (otherwise it is zero) and  $D_\lambda^-=\mathcal{O}(D^4)$ \cite{zhu2016clifford}. For the ir-rep where $D_\lambda^+=\tr(QP_\lambda)=0$, that is, $QP_\lambda=0$, one has $\Theta_1(\lambda)=\Theta_2(\lambda)=0$.
Recall Eq.~\eqref{eq:L0Cliff}, one finally has the contribution of $\Lambda_n^1$ as 
\begin{equation}\label{eq:ClL1}
    \begin{aligned}
&(D+1)^2\tr[\sigma\otimes O_0 \otimes \sigma\otimes O_0\  \Phi^{(4,\mathrm{Cl})}_n\left(\Lambda_n^1\right)]\\
=&(D+1)^2\tr\left\{\sigma\otimes O_0 \otimes \sigma\otimes O_0\ \left[
\sum_{\lambda} \Theta_1(\lambda) QP_{r}P_{\lambda}+\Theta_2(\lambda) QP_{r'}P_{\lambda}+\Theta_3(\lambda) P_{r}P_{\lambda}+\Theta_4(\lambda) P_{r'}P_{\lambda}\right]\right\}\\
\leq &(D+1)^2 \left[\sum_\lambda \Theta_1(\lambda)+\Theta_2(\lambda)\right] D^{-1}\tr(O_0^2)+(D+1)^2\left[\sum_\lambda\Theta_3(\lambda)+|\Theta_4(\lambda)|\right]\mathcal{O}(1)\tr(O_0^2)\\
=&\mathcal{O}(1)\tr(O_0^2).
\end{aligned}
\end{equation}
Here in the last but one line,  we apply Lemma \ref{prop:piO2Qnew} for the first two terms, and the last two terms can be bounded by following Eq.~\eqref{eq:Haaruselatter} in the Haar case. Combining Eq.~\eqref{eq:ClL0} and \eqref{eq:ClL1}, we finally prove Proposition \ref{prop:Cl} in main text.

\section{Numerical simulations and practical time-cost analysis}\label{AppendNum}
\subsection{Details of numerical simulations}\label{AppendNum1}

In this section, we clarify the technique and details of the numerical simulations for the scaling of statistical variance with shot-number $K$, setting-number $M$, weight $w$ of Pauli observables, and qubit number $n$ in Fig.~\ref{fig:Pauli-numeric} and \ref{fig:clifford-numeric}.

We first generate the processed state $\rho$, which is either a $GHZ$ state or a more general $GHZ_\theta$ state in quantum simulators. Then we perform randomized Pauli or Clifford measurements on $\rho$ for $M$ different random bases, each one repeating $K$ shots, and record the measurement results and the corresponding bases. Using the numerical results, we can construct the classical shadow $\hat{\rho}$ through Eq.~\eqref{eq:shadowPost} and Eq.~\eqref{rhoij}, and calculate the value of $\hat{o}=\tr(O\hat{\rho})$. Here, the observables $O$ are set as illustrated in the captions of Fig.~\ref{fig:Pauli-numeric} and \ref{fig:clifford-numeric}. Finally, for each pair of $(M,K)$ which generates a total estimator $\hat{\rho}$ in Eq.~\eqref{rhoij}, we repeat the whole randomized measurement procedure for $T=2000$ times to obtain the estimation of the statistical variance, with $\mathrm{Var}\left[\tr(O \widehat{\rho})\right]\sim \frac{1}{T}\sum^{T}_{t=1}\left[\tr(O \widehat{\rho_{t}}) - \tr(O\rho) \right]^{2}$. 

The numerical simulation results of random Pauli and Clifford measurements in Fig.~\ref{fig:Pauli-numeric} and \ref{fig:clifford-numeric} show consistency with the analytical results given by Eq.~\eqref{eq:pauli}, \eqref{eq:pauli-seperate}, and \eqref{eq:Cl:var}.

\subsection{Practical time-cost and accuracy analysis for random Pauli measurements}
Here, we illustrate the advantage of the introduced multi-shot protocol, by investigating the variance and time cost enhancements against the original shadow protocol for fixed measurement budgets and accuracy requirements. 

Generally, the total experimental time cost $\mathcal{C}_{tot}$ depends on both the time to change measurement settings $\mathcal{C}_{M}$ and the time to perform single-shot measurements $\mathcal{C}_{K}$. We assume, similar as in Ref.~\cite{seif2023shadow}, $\mathcal{C}_{M}=1000$ and $\mathcal{C}_{K}=1$, which indicates that changing the measurement settings requires more time than performing measurements in a fixed basis. Thus, the measurement budget is approximately $\mathcal{C}_{tot} = 1000 M + MK$.
As indicated in Eq.~\eqref{eq:pauli}, the variance to measure a Pauli operator $P$ using random Pauli measurements depends on $\{M, K, w, \tr(P\rho)\}$. Here, we directly set $w=2$, $\tr(P\rho)=0.2$ and explore the scaling of the variance and total time cost with the shot-number $K$ and setting-number $M$, as shown in Fig.~\ref{fig:time-cost}.

\begin{figure}[h]
\centering
\includegraphics[width=0.65\textwidth]{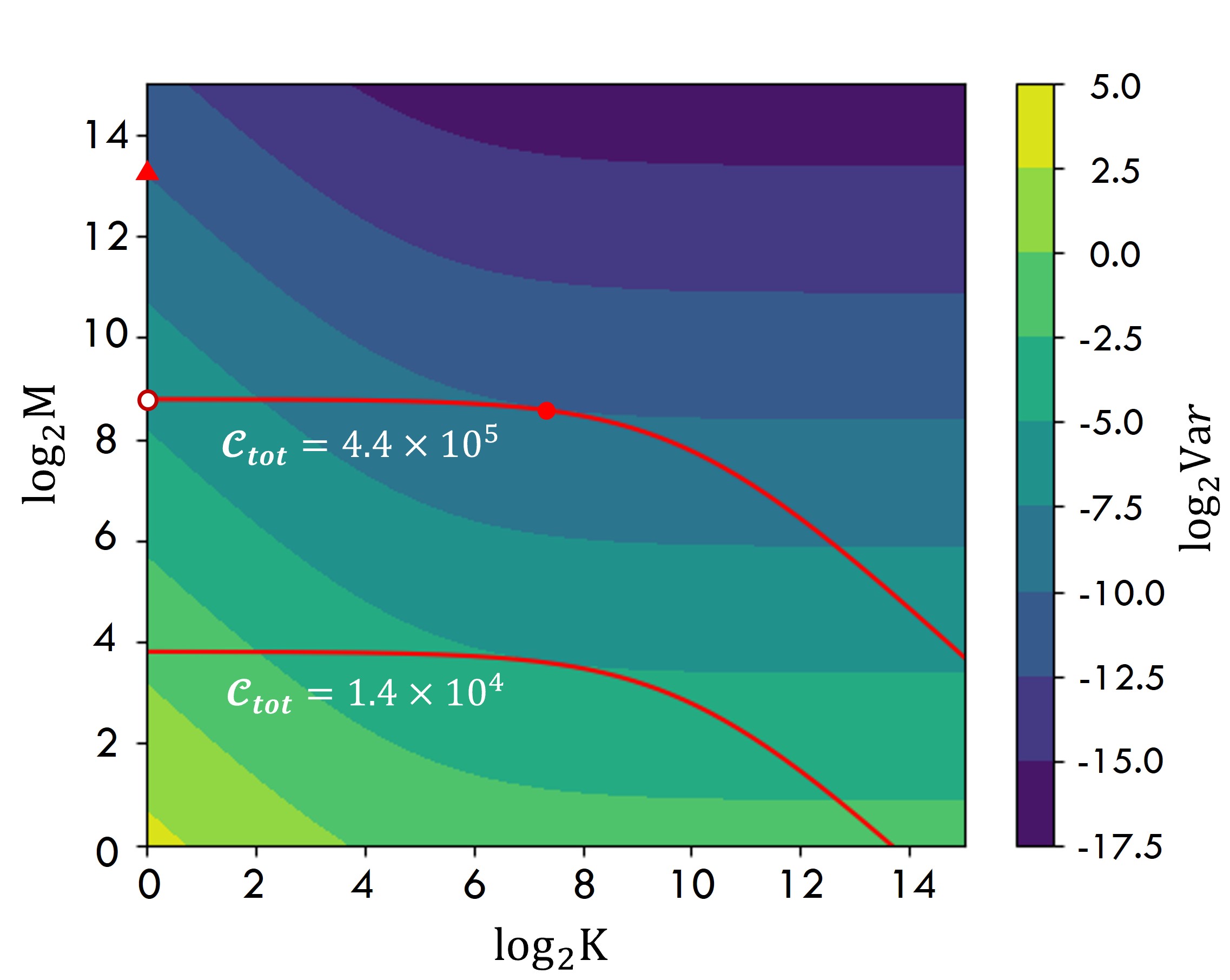}
\caption{Scaling of the analytical variance of random Pauli measurements and total time cost with $M$ and $K$. The variance (colored area) is given by Eq.~\eqref{eq:pauli} with $w=2$ and $\tr(P\rho)=0.2$.
The red lines show the contours of fixed $\mathcal{C}_{tot} = 1000 M + MK$ with different measurement budgets $\mathcal{C}_{tot}$. The red dot represents the optimal choice of $(M, K)$ for multi-shot protocol when $\mathcal{C}_{tot}=4.4\times 10^{5}$. The red ring on the vertical axis (i.e., $K=1$) is the choice of $M$ for the original shadow protocol using the same measurement budget, as that of the red dot. The red triangle on the vertical axis labels the corresponding choice of $M$ with the same variance as the red dot.
}
\label{fig:time-cost}
\end{figure}

Consider a fixed measurement budget $\mathcal{C}_{tot}=4.4\times 10^{5}$, we find the optimal choice of $(M, K)$ using our multi-shot protocol. It provides the lowest variance for estimating $\tr(P\rho)$, that is, $\mathrm{Var}_{ms}=0.001$, with the optimal choice approximately $M=380$ and $K=160$. Under the same budget, as shown in Fig.~\ref{fig:time-cost}, the variance using the original shadow protocol $\mathrm{Var}_{os}$ is much larger than using the multi-shot method $\mathrm{Var}_{ms}$, i.e., $\mathrm{Var}_{os}=0.011 > \mathrm{Var}_{ms}=0.001$. 
In the meantime, if one fixes the variance requirement as $\mathrm{Var}=0.001$, the total experimental time cost of the original shadow is $\mathcal{C}'_{tot} = 9.15\times 10^{6}$ marked by the point of red triangle in Fig.~\ref{fig:time-cost}, which is much larger than the one using our multi-shot technique.
Therefore, the results of the example above show advantages of our multi-shot shadow protocol against the original shadow from both the measurement budgets and the accuracy aspects.

\end{appendix}

\end{document}